\documentclass{article}
\usepackage[latin1]{inputenc}
\usepackage[english]{babel}%english
\usepackage{dsfont}
\usepackage[cyr]{aeguill}
\usepackage{amsmath,amssymb}
\usepackage{amsthm}
\usepackage{eurosym}
\usepackage{multicol}
\usepackage{color}
\usepackage{fancyhdr}
\usepackage[pdftex]{graphicx}
\usepackage[pdftex]{hyperref}
\usepackage{verbatim}
\usepackage{natbib}
\usepackage{setspace}
\usepackage{multirow}
\usepackage{stmaryrd}
\usepackage{bm} %pour mettre en gras les lettres grecques
\newcommand{\indic}{\mathds{1}} %indicatrice
\newcommand{\E}{\mathbb{E}} %esperance
\DeclareMathOperator{\Cov}{Cov} %covariance
\DeclareMathOperator{\tr}{tr} %trace

\DeclareMathOperator{\argmax}{argmax}

\theoremstyle{plain}
\newtheorem{thm}{Theorem}

\newtheorem{lem}{Lemma}

\providecommand{\keywords}[1]{\textbf{\textit{Keywords --}} #1}

\title{Composite likelihood inference of fractional Gaussian processes with sequentially optimal subset selection}%{Randomized composite marginal likelihood approach for the inference of multivariate fractional Brownian motions with application to rough volatilities}

\author{Mathis Fourreau$^{\text{a}}$, Matthieu Garcin$^{\text{b,c,}}$\footnote{Corresponding author: matthieu.garcin@m4x.org. \newline $^{\text{a}}$ ESILV, 92916 Paris La Défense, France. \newline $^{\text{b}}$ De Vinci Higher Education, De Vinci Research Center, Paris, France. \newline $^{\text{c}}$ This research benefited from the support of the Chaire ``Deep Finance Statistics'' between QRT, Ecole Polytechnique and its foundation.
\newline Acknowledgment: MG is grateful to Thomas Savini for his assistance in the initial steps of this project.}}

\date{\today}

\begin{document}

\maketitle

\begin{abstract}
The composite likelihood method reduces the computational cost of parameter estimation in time series by considering several subsets of observations instead of all observations at once. The asymptotic properties of this method are related to the Godambe information, an extension of the Fisher information that accounts for the dependence between subsets of observations. We aim to apply this method to linear Gaussian models, in particular fractional Brownian motion and fractional Gaussian noise. We derive theoretical expressions for their Fisher information and their Godambe information and deduce a subset selection design that sequentially maximizes the Godambe information. The size of the subsets then allows us to control the trade-off between estimation accuracy and computational cost. Through simulations, we compare this method with the method of moments and maximum likelihood estimation, and we apply it to real data, namely volatility series of a stock index and a wind speed time series.
\end{abstract}
%moins de 100 mots 

%\vspace{0.3cm}
%five or six keywords or key phrases, arranged in alphabetical order

%\textbf{\textit{MSC codes}} -- 60G10, 60G15, 60G18, 60G22, 62M09, 62M10, 62P05

%\textbf{\textit{JEL codes}} -- C52, C58

\section{Introduction}

By replacing the full log-likelihood by a sum of lower-dimensional marginal log-likelihoods, the composite likelihood approach provides an efficient framework for parameter inference in a high-dimensional model~\cite{Lindsay,VRF}. It is particularly useful for time series exhibiting non-trivial serial dependence~\cite{NJK}. One of the challenges associated with this method lies in how to select subsets of observations that achieve a good trade-off between computational cost and estimation accuracy~\cite{LYS}. Research on this topic remains in its early stages, despite several interesting contributions~\cite{PM,MKR,HF}.

To address this question of optimal design, we study a simple time series model that is Gaussian and whose dependence structure is governed by a single parameter. We therefore focus our analysis on fractional Gaussian processes, which are widely used in many fields, including finance and meteorology. Depending on the value of the Hurst parameter, these processes may exhibit long-range dependence, characterized by slowly decaying correlations. However, when long sequences of discrete-time observations are available, maximum likelihood estimation becomes computationally prohibitive, as it requires inverting the autocovariance matrix and computing its determinant~\cite{Coeurjolly2000,GarcinEstimLamp}. As a result, many alternative methods, which are less accurate but computationally faster, have been developed~\cite{Coeurjolly2000,JHJ}, including method-of-moments estimators ~\cite{Coeurjolly2001,Garcin2017} or the Whittle estimator, which approximates the likelihood in the frequency domain~\cite{Beran,Coeurjolly2000,SYZ}.

We propose a composite likelihood approach to estimate the Hurst exponent of fractional Gaussian processes and study the effect of subset design on the resulting Godambe information. As one of our contributions, we derive a theoretical expression for this information, which is then used to construct a subset selection procedure. Due to combinatorial complexity, we adopt a sequential selection strategy in which, at each step, a subset of size $p$ is chosen to maximize the Godambe information conditionally on the previously selected subsets. The tuning parameter $p$ makes it possible to control the trade-off between computational cost and estimation accuracy; in particular the performance of the standard maximum likelihood estimator is recovered when $p$ equals the total number of observations. It is also possible to determine the minimal value of $p$ for which the proposed estimator outperforms the moment-based approach in terms of accuracy.

The rest of the paper is organized as follows. Section~\ref{sec:defstandardestim} briefly introduces fractional Gaussian models along with the benchmark estimation methods considered in this work, namely the maximum likelihood estimation and a moment-based approach. A comparison with Whittle estimation is beyond the scope of the present work, since our goal is not to compare spectral and time-domain methods, but rather to investigate the behavior of composite likelihood constructions in the time domain. The composite likelihood approach is presented and applied to fractional Gaussian processes in Section~\ref{sec:composite}. We also discuss the resulting design in this section. Sections~\ref{sec:simul} and~\ref{sec:appli} are respectively devoted to a simulation study and to a short application to real financial and climatic data, mainly intended to illustrate the practical relevance of the proposed approach. Section~\ref{sec:concl} concludes.

\section{Definition and standard estimation of the fBm and of the fGn}\label{sec:defstandardestim}

\subsection{Definition of fractional Gaussian processes}

Given the Hurst parameter $H$, the fractional Brownian motion (fBm) $X^H_t$ is defined, when $H\in(1/2,1)$ (respectively $(0,1/2)$), as a fractional primitive (resp. derivative) of a standard Brownian motion, thus leading to a integral expression~\cite{MvN}. When $H=1/2$, we recover the Brownian motion. When working with discrete-time observations, it is however simpler to use an alternative definition. Indeed, a vector of observations is then a centered Gaussian vector, characterized by its autocovariance matrix, which is based on
\begin{equation}\label{eq:covFBM}
\E\left[X^H_sX^H_t\right]=\frac{\sigma^2}{2}\left(|s|^{2H}+|t|^{2H}-|t-s|^{2H}\right),
\end{equation}
where $s,t\in\mathbb R$ and $\sigma^2$ is a scale parameter, equal to 1 in the standard specification.

The fractional Gaussian noise (fGn) is defined as an increment of an fBm on a time step equal to 1: $Y^H_t=X^H_{t+1}-X^H_t$. It is thus easy to deduce the autocovariance of this process from equation~\eqref{eq:covFBM}:
\begin{equation}\label{eq:covFGN}
\E\left[Y^H_sY^H_t\right]=\frac{\sigma^2}{2}\left(|t-s-1|^{2H}-2|t-s|^{2H}+|t-s+1|^{2H}\right).
\end{equation}

After equations~\eqref{eq:covFBM} and~\eqref{eq:covFGN}, we conclude that the fGn only is stationary. It can also be proved that the fBm only is self-similar. Moreover, for the two models, the entire serial dependence structure is encoded in a single parameter, $H$~\cite{BG}. We therefore focus on the estimation of $H$ alone, rather than on the joint estimation of $(H,\sigma^2)$.

\subsection{Moment-based inference method for fractional Gaussian processes}\label{sec:moment}

Let $\textbf{V}=(X_1,...,X_N)'$ be a vector of observations, whose statistical model is an fBm of Hurst exponent $H$, and where we use the notation $M'$ for the transpose matrix of $M$.

For a given time scale $m\in\llbracket 0,M\rrbracket$, with $M\leq N-1$, the variance of an increment of an fBm is $\sigma^2 m^{2H}$. Noting $\mathcal E(m)=\sum_{i=1}^{N-m}(X_{i+m}-X_i)^2/(N-m)$ the empirical variance, the moment-based estimator of $H$ is obtained as the slope of the linear regression of $\{\log\mathcal E(m)\}_{1\leq m\leq M}$ on $\{2\log m\}_{1\leq m\leq M}$. This estimator is strongly consistent and asymptotically normal when $H<3/4$~\cite{IL,Coeurjolly2001}. 

When $H$ is larger than $3/4$, the process exhibits stronger longe-range dependence and a second-order differencing scheme is required to mitigate it and recover standard asymptotic properties. We therefore define $\mathcal E_2(m)=\sum_{i=1}^{N-2m}(X_{i+2m}-2X_{i+m}+X_i)^2/(N-2m)$ and the new estimator of $H$ is the slope of the linear regression of $\{\log\mathcal E_2(m)\}_{1\leq m\leq M}$ on $\{2\log m\}_{1\leq m\leq M}$~\cite{IL,Coeurjolly2001}. This second-order estimator has however a lower accuracy than its first-order counterpart when $H$ is smaller than about 0.7~\cite[Section 5]{Coeurjolly2001}.

When the statistical model is an fGn, the simplest approach is to build first an fBm as a cumulative sum of the observations and to apply then the above method.

In what follows, we have implemented the method with a first-order difference scheme, with $M=5$, in line with recommendations from the literature~\cite[Section 5.2]{Coeurjolly2001}.

\subsection{Standard maximum likelihood approach for fractional Gaussian processes}\label{sec:mle}

We still consider the vector of observations $\textbf{V}=(X_1,...,X_N)'$ with an fBm of Hurst exponent $H$ as a statistical model. Its log-likelihood in the fBm framework is:
\begin{equation}\label{eq:loglik}
\ell (\textbf V;H,\sigma^2)=-\frac{1}{2}\ln\left(\det\left[\sigma^2\bm\Sigma\right]\right)-\frac{N}{2}\ln(2\pi)-\frac{1}{2\sigma^2}\textbf{V}'\bm\Sigma^{-1}\textbf{V},
\end{equation}
where $\bm\Sigma$ is the theoretical covariance matrix of an fBm of Hurst exponent $H$ and unit scale parameter, based on equation~\eqref{eq:covFBM}. Introducing the differentiation operator $\partial_{x}$ with respect to a variable $x$, we have 
$$\partial_{\sigma^2}\ell (\textbf V;H,\sigma^2) = -\frac{N}{2\sigma^2}+ \frac{1}{2\sigma^4}\textbf{V}'\bm\Sigma^{-1}\textbf{V},$$
which is equal to zero for $\sigma^2=\widehat{\sigma}^2_{\ell}$, the maximum likelihood estimator:
$$\widehat{\sigma}^2_{\ell}=\frac{1}{N}\textbf{V}'\bm\Sigma^{-1}\textbf{V}.$$
We can then plug this estimator in equation~\eqref{eq:loglik}, which now depends on the data $\textbf V$ and on the only parameter $H$:
\begin{equation}\label{eq:llvect}
\ell (\textbf V;H)=-\frac{N}{2}\ln\left(\frac{\textbf{V}'\bm\Sigma^{-1}\textbf{V}}{N}\right) -\frac{1}{2}\ln\left(\det\left[\bm\Sigma\right]\right)-\frac{N}{2}\left(\ln(2\pi)+1\right).
\end{equation}
This plug-in makes it possible to estimate $H$ easily as the solution of a one-dimensional optimization problem:
$$\widehat H_{\ell} = \underset{H\in(0,1)}{\arg\max}\ \ell (\textbf V;H).$$

We note that if we use instead the statistical model of an fGn, the above equations still apply if one replaces the matrix $\bm\Sigma$ by $\bm\Gamma$, the theoretical covariance matrix of an fGn of Hurst exponent $H$ and unit variance parameter, determined with the help of equation~\eqref{eq:covFGN}.

Whatever the above Gaussian model, basic algorithms for calculating the likelihood have a complexity which is cubic in the number of observations $N$. This stems from both the matrix inversion and the computation of the determinant.

Noting that $\bm\Gamma$ is Toeplitz, Hermitian, and positive definite, some faster algorithms exist for matrix inversion and determinant calculation in the fGn model~\cite{Coeurjolly2000}. Regarding the matrix inversion, the Trench algorithm, a recursive procedure based on a block decomposition, has a complexity in $\mathcal O(N^2)$~\cite{Trench,Zohar}. The calculation of the determinant also reaches a quadratic complexity~\cite{Cinkir} but, if done jointly with the matrix inversion, it only requires $N$ additional multiplications~\cite{Trench,Zohar}. The Trench algorithm thus leads to a faster computation of the likelihood in the case of an fGn. In the numerical experiments of this paper, we will use this algorithm to build a benchmark estimator.

Though the complexity of the Trench algorithm is lower than the one of basic algorithms, its multiple use in the maximization of the likelihood leads to prohibitive calculation times as soon as one studies time series of several thousand observations. One might thus be tempted to use superfast Toeplitz inversion methods. These methods, often based on a fast Fourier transform, make it possible to reach a complexity in $\mathcal O(N\log^2N)$. However, contrary to the Trench algorithm, they only provide an approximation of the inverse matrix and they are known to be numerically unstable~\cite{AG,KMS,PT}.

The fast algorithms detailed above require the matrix to be Toeplitz, which is the property of $\bm\Gamma$ but not of $\bm\Sigma$. In other words, the series of observations of the fBm is to be differentiated first in order to work with an fGn. Observations of an fBm or an fGn not equally distributed in time would also constrain us to use only basic and slow algorithms.

\section{Composite marginal likelihood method}\label{sec:composite}

In what follows, we first present the general outline of the composite likelihood method, in order to fix the notation, and we then apply it to fBm and fGn, deriving the Fisher information and the Godambe information. We finally draw numerical conclusions from these theoretical developments to properly select the subsets of observations to be used in the composite likelihood and address as well the specific question of the time scale at which the fGn is observed. 

\subsection{Method outline}\label{sec:MethodOutline}

In this paragraph, we place ourselves in a general framework and we will apply the method to fBm and fGn afterwards. We thus summarize below the main elements of the composite likelihood approach. More details can be found in some well-known papers on the subject~\cite{LYS,VRF}.

Let $\textbf{V}$ be a vector of $N$ observations, following a statistical model of parameter $\bm\theta\in\Theta$, where $\Theta$ is a subset of $\mathbb R^d$. We select a number $N_{\text{cl}}$ of classes, each consisting in a subset of observations of size $p$, contained in the vector $\textbf{V}_k=\mathbf{S}_k\textbf{V}$ for $k\in\llbracket  1,N_{\text{cl}}\rrbracket$, where $\mathbf{S}_k$ is a binary matrix of size $p\times N$. If $\mathbf V_k$ contains the observations of index $k_1,...,k_p$, the element of $\mathbf{S}_k$ in position $(i,j)\in \llbracket 1,p\rrbracket\times\llbracket 1,N\rrbracket$ is $\indic_{j=k_i}$. The marginal composite log-likelihood of $\textbf{V}$, based on the sub-vectors $\textbf{V}_1,...,\textbf{V}_{N_{\text{cl}}}$, is defined by:
$$\mathcal C(\bm\theta)=\sum_{k=1}^{N_{\text{cl}}}\ell(\textbf{V}_k;\bm\theta),$$
where $\ell(\textbf{V}_k;\bm\theta)$ is the log-likelihood of the sub-vector $\textbf{V}_k$.

For each sub-vector $\textbf{V}_k$, one determines the score $\mathbf u(\textbf{V}_k;\bm\theta)=\nabla \ell(\textbf{V}_k;\bm\theta)$, where $\nabla$ is the gradient operator with respect to $\bm\theta$. By stacking all the scores together, we obtain the function $u:\bm\theta\mapsto(\mathbf u(\textbf{V}_1;\bm\theta)',...,\mathbf u(\textbf{V}_{N_{\text{cl}}};\bm\theta)')'$, which outputs a vector $\mathbf{u}(\bm\theta)$ of size $dN_{\text{cl}}$. One then considers a weighted average of the coordinates of $\mathbf{u}(\bm\theta)$, $\mathbf{\mathcal U}(\bm\theta)=\mathbf{W}(\bm\theta)\mathbf{u}(\bm\theta)$, where the weight matrix $\mathbf{W}(\bm\theta)$ is of dimension $d\times dN_{\text{cl}}$ and does not depend on the data. In the case of equally-weighted sub-scores, the weight is simply $\mathbf W(\bm\theta)=\mathbf W_{I}$, defined by a number $N_{\text{cl}}$ of $d\times d$ identity matrices placed side by side. This specific choice of weight matrix leads to a score which derives naturally from the composite likelihood because $\mathbf W_{I}\mathbf u(\bm\theta)=\nabla\mathcal C(\bm\theta)$.

Like the scores, it is possible to calculate the Fisher information matrix for each sub-vector, $\mathbf{\mathcal I}(\textbf{V}_k;\bm\theta)=\E\left[\mathbf u(\textbf{V}_k;\bm\theta)\mathbf u(\textbf{V}_k;\bm\theta)'\right]$, and to stack these matrices in the following matrix of size $dN_{\text{cl}}\times d$:
$$\mathcal F(\bm\theta)=\left(\begin{array}{c}
\mathbf{\mathcal I}(\textbf{V}_k;\bm\theta) \\
\vdots \\
\mathbf{\mathcal I}(\textbf{V}_{N_{\text{cl}}};\bm\theta)
\end{array}\right).$$
The Fisher information of each sub-vector can also be found in the diagonal $d\times d$ blocks of the covariance matrix of the stacked scores, $\mathcal V(\bm\theta)=\E\left(\mathbf u(\bm\theta)\mathbf u(\bm\theta)'\right)$, which is of size $dN_{\text{cl}}\times dN_{\text{cl}}$. In the case where the sub-scores are uncorrelated, the off-diagonal blocks of the matrix $\mathcal V(\bm\theta)$ are null matrices. 

However, in general, scores of different sub-vectors are correlated and one uses the Godambe information instead of the Fisher information. Given the weight matrix $\mathbf W(\bm\theta)$, the Godambe information is defined by the following $d\times d$ matrix:
$$\mathcal J_{\mathbf W(\bm\theta)}((\textbf{V}_k)_{k\in\llbracket  1,N_{\text{cl}}\rrbracket};\bm\theta) = \mathbf W(\bm\theta)\mathcal F(\bm\theta)\left[\mathbf W(\bm\theta)\mathcal V(\bm\theta)\mathbf W(\bm\theta)'\right]^{-1}\left[\mathbf W(\bm\theta)\mathcal F(\bm\theta)\right]'.$$
In what follows, we will always use $\mathbf W(\bm\theta)=\mathbf W_{I}$ and we will simply use the notation  $\mathcal J$ for $\mathcal J_{\mathbf W_I}$. We will also omit the argument $(\textbf{V}_k)_{k\in\llbracket  1,N_{\text{cl}}\rrbracket}$ in $\mathcal J$ when not necessary for understanding. 

The usefulness of the Godambe information instead of the traditional Fisher information appears in the asymptotic theory related to composite likelihoods. For a vector of $N$ observations $\textbf V$, the composite maximum likelihood estimator of $\bm\theta$ is
$$\widehat {\bm\theta}_{\mathcal C} = \underset{\bm\theta\in\Theta}{\arg\max}\ \mathcal C(\bm\theta).$$
If we now suppose we are given $n$ independent copies of $\textbf V$, the composite maximum likelihood estimator $\widehat{\bm\theta}_{\mathcal C,n}$ is the $\bm\theta$ maximizing the sum of the composite likelihoods over all the $n$ independent vectors. Then, $\widehat{\bm\theta}_{\mathcal C,n}$ is asymptotically normal with a variance equal to the inverse of the Godambe information $\mathcal J(\bm\theta)$ related to the $N_{\text{cl}}$ sub-vectors of $\textbf V$:
\begin{equation}\label{eq:CLT_Godambe}
\sqrt{n}\left(\widehat{\bm\theta}_{\mathcal C,n} -\bm\theta\right)\rightsquigarrow \mathcal N_d\left(0,\mathcal J(\bm\theta)^{-1}\right),
\end{equation}
where $\rightsquigarrow$ stands for the convergence in distribution, $\mathcal N_d$ for the $d$-dimensional Gaussian distribution with mean and variance as parameters~\cite{VRF}. In practical applications, we will often have $n=1$. Therefore, even if the size $N$ of $\mathbf V$ is very high, the above asymptotic framework does not apply. However, for a fixed size $p$ of each sub-vector and $n=1$, when the number of sub-vectors $N_{\text{cl}}$ and thus the number of observations $N$ tend to infinity, the estimator $\widehat {\bm\theta}_{\mathcal C}$ may asymptotically follow a Gaussian distribution, provided that the dependence between the scores of the sub-vectors are not too strongly dependent to each other~\cite{VRF,BeG}. In that case, the inverse Godambe information also plays an important role in the asymptotic variance of the estimator. This will be confirmed in Section~\ref{sec:Design}, in the case of fGn, for the selected design. Moreover, in the non-asymptotic framework, which is of practical interest and which we study in the simulation study in Section~\ref{sec:simul}, the Godambe information is a reasonable criterion to be maximized, when selecting sub-vectors in the definition of the composite likelihood.

Indeed, two particular situations illustrate the intuition behind the Godambe information. First, when the scores of different sub-vectors are all uncorrelated, the Godambe information is equal to the sum of the Fisher information measures over all the sub-vectors. Maximizing the Godambe information in this case is equivalent to maximizing the average Fisher information over all the sub-vectors of observations. Second, when the scores are correlated and $d=1$, which is the framework later deepened in this article with the fBm and fGn models, the Godambe information is:
\begin{equation}\label{eq:GodambeUnSeulParam}
\mathcal J((\textbf{V}_k)_{k\in\llbracket  1,N_{\text{cl}}\rrbracket};\bm\theta) = \frac{\left[\sum_{k=1}^{N_{\text{cl}}}\mathcal I(\textbf{V}_k;\bm\theta)\right]^2}{\sum_{k=1}^{N_{\text{cl}}}\sum_{j=1}^{N_{\text{cl}}}\E\left(\mathbf u(\textbf{V}_k;\bm\theta)\mathbf u(\textbf{V}_j;\bm\theta)\right)}.
\end{equation}
Starting from an arbitrary set of sub-vectors of observations, if one wants to select an additional sub-vector maximizing the Godambe information, equation~\eqref{eq:GodambeUnSeulParam} shows that this sub-vector should have a high intrinsic Fisher information together with a score weakly correlated with the ones of the initial sub-vectors. 

\subsection{Information for fractional Gaussian processes}\label{sec:infoFG_Gaussian}

\subsubsection{Fisher information}\label{sec:FisherGaussian}

We are interested in the expression of the Fisher information of a sub-vector $\mathbf{V}_k=(Y_{t_{k,1}},...,Y_{t_{k,p}})'$ following a Gaussian distribution of theoretical covariance matrix $\sigma^2 \mathbf R_k$, where $\mathbf R_k$ depends on a single parameter $H$. As a particular case, we can apply this expression to an fBm (respectively an fGn), in which case $\mathbf R_k=\bm \Sigma_k$ (resp. $\bm \Gamma_k$), which is the theoretical covarance matrix of a vector of observations of an fBm (resp. fGn) with unit scale at times $t_{k,1},...,t_{k,p}$.

The likelihood of $\mathbf{V}_k$ follows equation~\eqref{eq:loglik}, with $\mathbf R_k$ instead of $\bm \Sigma$, and its score is 
\begin{equation}\label{eq:scoreTrace}
\partial_H\ell(\mathbf{V}_k;H,\sigma^2)=-\frac{1}{2}\partial_H\ln\left(\det(\mathbf R_k)\right)-\frac{1}{2\sigma^2}\mathbf{V}_k'\partial_H\mathbf R_k^{-1}\mathbf{V}_k.
\end{equation}
Equation~\eqref{eq:llvect} provides an alternative expression of the likelihood without the parameter $\sigma^2$ but the expression of the score deduced from this equation is much more complex. Therefore, we focus on the score obtained in equation~\eqref{eq:scoreTrace}, expression that we simplify below. First, we know that
$$\left(\partial_H\mathbf R_k^{-1}\right)\mathbf R_k=\partial_H(\mathbf R_k^{-1}\mathbf R_k)-\mathbf R_k^{-1}\partial_H\mathbf R_k=-\mathbf R_k^{-1}\partial_H\mathbf R_k,$$
so that
\begin{equation}\label{eq:IPPscore}
\partial_H\mathbf R_k^{-1}=-\mathbf R_k^{-1}\mathbf R_{k,H}\mathbf R_k^{-1},
\end{equation}
where $\mathbf R_{k,H}$ is the derivative of $\mathbf R_k$ with respect to $H$. Second, we note that, for any square matrix $\mathbf M$ of size $p\times p$, we have:
\begin{equation}\label{eq:scoreVMV}
\E\left[\mathbf{V}_k'\mathbf M\mathbf{V}_k\right]=\sum_{i=1}^p\sum_{j=1}^p (\mathbf M)_{ij}\E[Y_{t_i}Y_{t_j}]=\sum_{i=1}^p\sum_{j=1}^p (\mathbf M)_{ij}\sigma^2(\mathbf R_k)_{ji}=\sigma^2\tr(\mathbf M\mathbf R_k).
\end{equation}
Third, since the expected value of the score is zero, equations~\eqref{eq:scoreTrace}, \eqref{eq:IPPscore}, and~\eqref{eq:scoreVMV} provide:
$$\partial_H\ln\left(\det(\mathbf R_k)\right)=\frac{1}{\sigma^2}\E\left[\mathbf{V}_k'\mathbf R_k^{-1}\mathbf R_{k,H}\mathbf R_k^{-1}\mathbf{V}_k\right]=\tr(\mathbf R_k^{-1}\mathbf R_{k,H}).$$
Finally, the score follows
\begin{equation}\label{eq:scorefinal_trace}
\partial_H\ell(\mathbf{V}_k;H,\sigma^2) = -\frac{1}{2}\tr(\mathbf R_k^{-1}\mathbf R_{k,H})+\frac{1}{2\sigma^2}\mathbf{V}_k'\mathbf R_k^{-1}\mathbf R_{k,H}\mathbf R_k^{-1}\mathbf{V}_k.
\end{equation}
This result is evoked for instance in~\cite{SCA,ACS}, with elements of an alternative proof in~\cite[p308 and eq.8.6]{Harville}. We can now write the Fisher information for $H$:
\begin{equation}\label{eq:infoFisher_trace}
\begin{array}{ccl}
\mathcal I(\mathbf{V}_k;H) & = & -\E\left(\partial^2\ell(\mathbf{V}_k;H,\sigma^2)/\partial H^2\right) \\
& = & \frac{1}{2}\left[\partial_H \tr\left(\mathbf R_k^{-1}\mathbf R_{k,H}\right) - \frac{1}{\sigma^2}\E\left(\mathbf{V}_k'\partial_H\left(\mathbf R_k^{-1}\mathbf R_{k,H}\mathbf R_k^{-1}\right)\mathbf{V}_k\right)\right] \\
& = & \frac{1}{2}\left[\partial_H\tr\left(\mathbf R_k^{-1}\mathbf R_{k,H}\right) - \tr\left(\partial_H\left(\mathbf R_k^{-1}\mathbf R_{k,H}\mathbf R_k^{-1}\right)\mathbf R_k\right)\right] \\
& = & -\frac{1}{2}\tr\left(\mathbf R_k^{-1}\mathbf R_{k,H}\left(\partial_H\mathbf R_k^{-1}\right)\mathbf R_k\right) \\
& = & \frac{1}{2}\tr\left(\mathbf R_k^{-1}\mathbf R_{k,H}\mathbf R_k^{-1}\mathbf R_{k,H}\right),
\end{array}
\end{equation}
where we used equation~\eqref{eq:scorefinal_trace} in the second line, equation~\eqref{eq:scoreVMV} in the third line, an expansion and a simplification of the derivative of the product of $\mathbf R_k^{-1}\mathbf R_{k,H}$ with $\mathbf R_k^{-1}$ in the fourth line, and equation~\eqref{eq:IPPscore} in the last line. This simple result is also evoked without a proof in~\cite{AW}.

As an illustration, we now focus on the case where $p=2$. Theorem~\ref{thm:Fisher_fGn} provides a general expression for the Fisher information of a centered Gaussian pair with a correlation matrix which depends on a single parameter $H$. Then, assuming that the observed pair follows an fGn, we also have an approximation of this information.

\begin{thm}\label{thm:Fisher_fGn}
Let $t\in\mathbb R$, $\tau>0$, and $(Y_{t}, Y_{t+\tau})$ be a pair of observations of a centered stationary Gaussian process, such that their correlation $\rho_{H,\tau}$ depends on a parameter $H\in\mathbb R$. The Fisher information of this vector writes
\begin{equation}\label{eq:Fisher_Gauss}
\mathcal I(Y_{t},Y_{t+\tau};H)=\frac{(1+\rho_{H,\tau}^2)(\partial_H\rho_{H,\tau})^2}{(1-\rho_{H,\tau}^2)^2}.
\end{equation}
When $(Y_t)_{t\in\mathbb R}$ is an fGn and $\tau>1$, the Fisher information can also be written as:
$$\mathcal I(Y_t, Y_{t+\tau};H) = \frac{1+a_{\tau}^2}{(1-a_{\tau}^2)^2}b_{\tau}^2 + \mathcal R_{\tau},$$
with $a_{\tau}=H(2H-1)\tau^{2H-2}$, $b_{\tau}=2a_{\tau}\ln \tau+(4H-1)\tau^{2H-2}$, and
$$|\mathcal R_{\tau}|
\le
\tau^{2H-4}
\Bigg[
\frac{1+q_{\tau}^2}{(1-q_{\tau}^2)^2}
\,M_{\tau}
\Big(2|b_{\tau}|+\tau^{2H-4}M_{\tau}\Big)
+
b_{\tau}^2\frac{\alpha_{g,3}q_{\tau}(3+q_{\tau}^2)}{24(1-q_{\tau}^2)^3}
\Bigg],$$
when $\tau$ and $H$ are such that $q_{\tau}<1$, where
$$\left\{\begin{array}{ccl}
q_{\tau} & = & \left(\tau^{2H-4}\alpha_{g,3}/48\right) + |a_{\tau}| \\
M_{\tau} & = & \left(\alpha_{g,3}\ln(\tau) + \alpha_{h,3}\right)/24 \\
\alpha_{g, n} & = & 2\left(1-\tau^{-1}\right)^{2H-n-1}|A_{n+1}| \\
\alpha_{h,n} & = & 2\left(1-\tau^{-1}\right)^{2H-n-1}\left(\left|A_{n+1}\ln{\left(1-\tau^{-1}\right)}\right|+|B_{n+1}|\right),
\end{array}\right.$$
with the sequences $A_n$ and $B_n$ defined by $A_0=1$, $B_0=0$, and
\begin{equation}\label{eq:FisherGauss_AB}
\left\{\begin{aligned}
A_{n+1} &= (2H-n)A_n \\
B_{n+1} &= (2H-n)B_n + A_n.
\end{aligned}\right.
\end{equation}
\end{thm}

The proof of Theorem~\ref{thm:Fisher_fGn} is postponed in Appendix~\ref{sec:proof_Fisher_fGn}.

Theorem~\ref{thm:Fisher_fGn} gives a general expression for the Fisher information. It does not depend on $\sigma^2$ but only on the correlation $\rho_{H,\tau}$ and on its derivative with respect to $H$, thus on the parameter $H$ and on the time lag $\tau$. In the case of an fGn, an explicit expression exists for these two terms, as it appears in the proof of the theorem, but Theorem~\ref{thm:Fisher_fGn} provides as well a simple approximation of the Fisher information, based on a Taylor expansion. This approximation makes interpretation easier. Indeed, when $\tau\rightarrow\infty$, the leading term of the Fisher information of the fGn is simply $b_{\tau}^2=(2H(2H-1)\ln(\tau)+4H-1)^2\tau^{4H-4}$. The information decreases as the two observations become further apart in time, with a faster decay for a smaller $H$.

Figure~\ref{fig:InfoPaireFgn} shows the Fisher information of a pair of observations of the fGn for a large range of values of $\tau$. We observe that the information is unbounded when $\tau\rightarrow 0$. It corresponds to the observation of the process in continuous time, a framework in which the estimation is much more accurate. But, in realistic applications, we only have access to discrete-time observations. If the sampling is defined by a minimal time step of 1, then the information is maximal for consecutive observations. For a finer sampling, the lag of 1 is a local maximum. Indeed, two consecutive observations of the fGn, distant of each other by $\tau<1$, would then correspond to overlapping increments of an fBm, meaning that they partly contain redundant information.

\begin{figure}[htbp]
	\centering
		\includegraphics[width=0.75\textwidth]{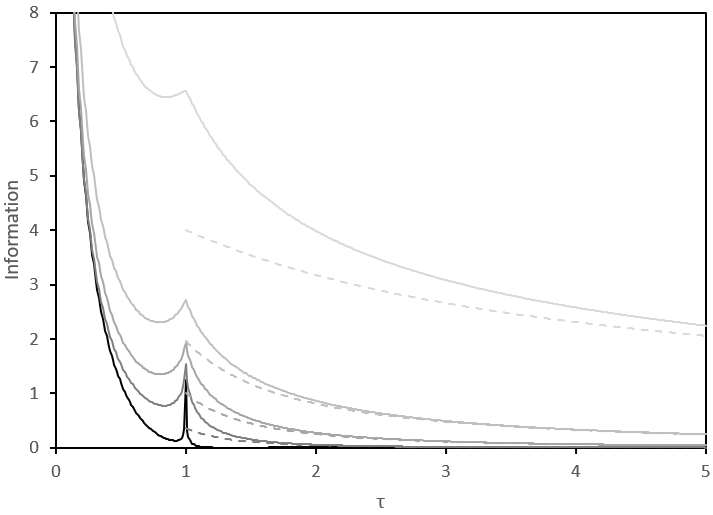} 
\begin{minipage}{0.7\textwidth}\caption{Fisher information for a pair of observations distant from each other by a duration $\tau$. The underlying model is an fGn of Hurst exponent equal to 0.25, 0.4, 0.5, 0.6, 0.75 (from the darker to the lighter curve). The dashed curves correspond to the asymptotic leading term $(2H(2H-1)\ln(\tau)+4H-1)^2\tau^{4H-4}$ obtained by Theorem~\ref{thm:Fisher_fGn}.}
	\label{fig:InfoPaireFgn}
\end{minipage}
\end{figure}

\subsubsection{Godambe information}

We apply the framework introduced in Section~\ref{sec:MethodOutline} to the case of a centered Gaussian process, characterized by a covariance matrix $\sigma^2\mathbf R$, where $\mathbf R$ depends on a single parameter $H$. From the vector of observations $\mathbf V$ of size $N$, we consider sub-vectors $\mathbf{V}_k=(Y_{t_{k,1}},...,Y_{t_{k,p}})'$, of covariance matrix $\sigma^2\mathbf R_k=\sigma^2\mathbf{S}_k\mathbf R\mathbf{S}_k'$, for $k\in\llbracket 1,N_{\text{cl}}\rrbracket$. 

We are interested in the estimation of $H$ only, and not in the estimation of $\sigma^2$. In order to determine the Godambe information associated with the set of sub-vectors and restricted to the estimation of $H$, we need the covariance between scores of all the sub-vectors. As scores, we in fact don't consider the score vectors but their entries related to the parameter $H$, namely $u(\mathbf V_k;H,\sigma^2)=\partial_H\ell(\mathbf{V}_k;H,\sigma^2)$. The parameter $\sigma^2$ appears in the expression of this entry but, as one can see in Theorem~\ref{thm:GodambeGauss}, it disappears in the expression of Godambe information.

\begin{thm}\label{thm:GodambeGauss}
The Godambe information of the centered Gaussian vectors $\mathbf V_1,...,\mathbf V_{N_{\text{cl}}}$ writes
$$\mathcal J((\mathbf V_k)_{k\in\llbracket 1,N_{\text{cl}}\rrbracket};H) = \frac{\left(\sum_{k=1}^{N_{\text{cl}}}\tr(\mathbf R_k^{-1}\mathbf R_{k,H}\mathbf R_k^{-1}\mathbf R_{k,H})\right)^2}{2\sum_{j=1}^{N_{\text{cl}}}\sum_{k=1}^{N_{\text{cl}}}\tr(\mathbf R_j^{-1}\mathbf R_{j,H}\mathbf R_j^{-1}\bm\Lambda_{j,k}\mathbf R_k^{-1}\mathbf R_{k,H}\mathbf R_k^{-1}\bm\Lambda_{k,j})},$$
where $\bm\Lambda_{j,k}=\mathbf{S}_j\mathbf R\mathbf{S}_k'$ and where $\mathbf R_{k,H}$ is the derivative of $\mathbf R_k$ with respect to $H$. In the case where the process is stationary, with a correlation function $\tau\mapsto\rho_{H,\tau}$, and where each sub-vector $\mathbf{V}_k = \left(Y_{t_k},\, Y_{t_k+\tau}\right)'$ is of size $p=2$, with $t_k\in\mathbb R$ and $\tau>0$, we have
\begin{equation}\label{eq:Godambe_p2}
\mathcal{J}\left((\mathbf{V_k})_{k \in  \llbracket 1, N_{cl}\rrbracket} ; H\right) =\frac{2{N_{\text{cl}}}^{2}(\partial_H\rho_{H,\tau})^{2}(1+\rho_{H,\tau}^{2})^2}{N_{\text{cl}} T_{0,H,\tau}+2 \sum_{1 \le k < j \le N_{\text{cl}}} T_{t_j-t_k,H,\tau}},
\end{equation}
with $T_{\theta,H,\tau}=2f_{H,\tau}(\theta,\theta+\tau)f_{H,\tau}(\theta,\theta-\tau)+f_{H,\tau}(\theta+\tau,\theta)^2+f_{H,\tau}(\theta-\tau,\theta)^2$, where $f_{H,\tau}(d_1,d_2)=-2\rho_{H,\tau}\rho_{H,d_1}+(1+\rho_{H,\tau}^2)\rho_{H,d_2}$.
\end{thm}

The proof of Theorem~\ref{thm:GodambeGauss} is postponed in Appendix~\ref{sec:proof_GodambeGauss}.

We note that, though $\mathbf R$ may be a huge matrix, as soon as $p$ is small, $\mathbf S_j\mathbf R \mathbf S_k'$ has a small size and is easy to calculate. Indeed, in practice, we don't interpret this formula as matrix products but as selection of particular entries of $\mathbf R$. In particular, when stacking $\mathbf V_1$ and $\mathbf V_2$, both of size $p$, one gets a vector of size $2p$ which has the following block covariance matrix:
$$\sigma^2\left(\begin{array}{cc}
\mathbf R_1 & \bm\Lambda_{1,2} \\
\bm\Lambda_{2,1} & \mathbf R_2
\end{array}\right)
=
\sigma^2\left(\begin{array}{cc}
\mathbf S_1\mathbf R \mathbf S_1' & \mathbf S_1\mathbf R \mathbf S_2' \\
\mathbf S_2\mathbf R \mathbf S_1' & \mathbf S_2\mathbf R \mathbf S_2'
\end{array}\right).$$

The performance of composite likelihood methods strongly depends on the choice of observation subsets. This topic is deepened in Section~\ref{sec:Design}, but it is worth noting that considering a very large number of subsets is not always the best design. Indeed, it has been shown that, in some cases, adding new subsets may lead to a less efficient estimator~\cite{XRX}. In a setting where the number $n$ of independent vectors of observations tends to infinity, such an increase of the variance of the estimator is equivalent to a decrease of the Godambe information, as stated in formula~\eqref{eq:CLT_Godambe}. For a single vector of $N$ observations, $(Y_0, Y_1,...,Y_{N-1})'$, and for data generated according to an fGn, we can also show that the Godambe information of a composite likelihood may decrease when one adds a new subset. Equation~\eqref{eq:GodambeUnSeulParam} clearly shows that this phenomenon arises when the additional subset contributes little marginal information at the numerator, while disproportionately increasing the denominator. In Figure~\ref{fig:GodambeDecr}, we focus on the case of a design by pairs, starting with 20 pairs $(Y_i,Y_{i+1})'$, for $i\in\llbracket 0,19\rrbracket$, and adding a new pair $(Y_{20},Y_{20+\tau})'$, with $\tau$ to be selected in $\llbracket -20,-2\rrbracket\cup\llbracket 1,N-21\rrbracket$. When the Hurst exponent $H$ is lower than 0.6, we observe that adding such a pair increases the Godambe information, whatever $\tau$. But for larger values of $H$, considering $\tau<0$ amounts to introducing a pair with a score much correlated to the one of the pairs already selected, so that the Godambe information, calculated after Theorem~\ref{thm:GodambeGauss}, may decrease. We also see in Figure~\ref{fig:GodambeDecr} that the maximum Godambe information is reached for $\tau=1$.

\begin{figure}[htbp]
	\centering
		\includegraphics[width=0.45\textwidth]{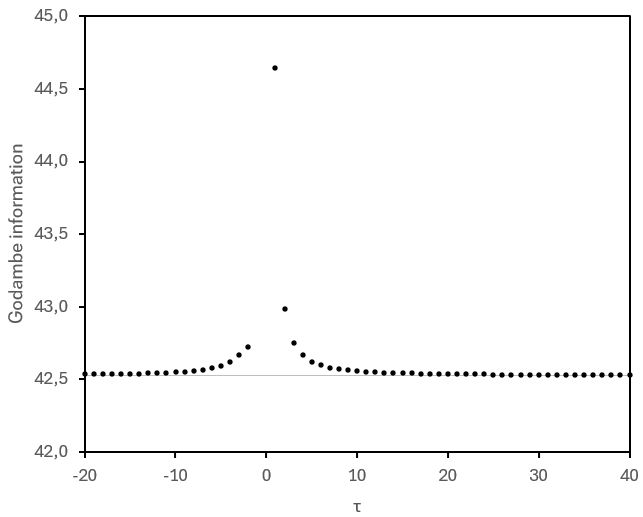}
		\includegraphics[width=0.45\textwidth]{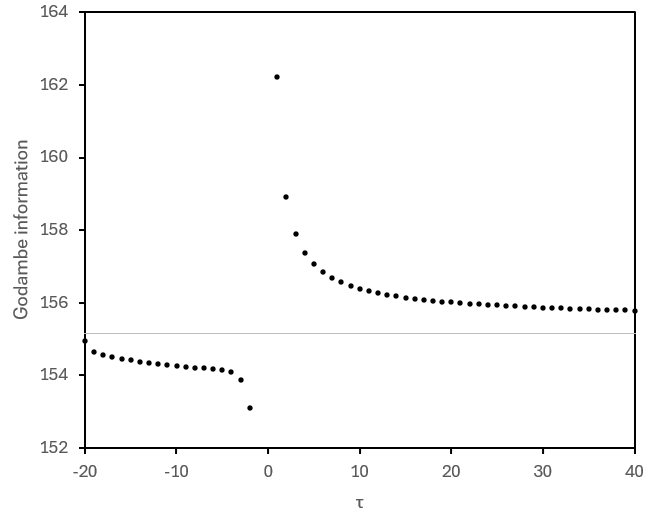} 
\begin{minipage}{0.7\textwidth}\caption{Godambe information of the pairs $(Y_i,Y_{i+1})'$, for $i\in\llbracket 0,19\rrbracket$, with (black dots) or without (grey line) the new pair $(Y_{20},Y_{20+\tau})'$, for $\tau\in\llbracket -20,-2\rrbracket\cup\llbracket 1,40\rrbracket$. The underlying model is an fGn of Hurst exponent equal to 0.55 (left graph) or 0.8 (right graph).}
	\label{fig:GodambeDecr}
\end{minipage}
\end{figure}

\subsection{Designing the subsets for fractional Gaussian processes}\label{sec:Design}

When using composite likelihood methods, a key step is the selection of the subsets of the $N$ observations on which the likelihood components are built. The choice of these subsets directly affects both the statistical efficiency of the estimator and the computational cost of the method. We restrict ourselves to subsets of common size $p$. The pairwise likelihood ($p=2$) is the most widespread choice, but considering higher values of $p$ generally improves the estimator.

The most natural design consists in considering all the subsets of size $p$. However, their number, $N!/p!(N-p)!$, grows rapidly with $N$, rendering the method computationally prohibitive.

Two main extensions are proposed in the literature. First, one can add weights to the components of the composite likelihood~\cite{VV,JL}. Optimal weights have been proposed but this method does not solve the problem of the computational cost~\cite{LYS}. Second, one can select fewer sub-vectors of size $p$, according to some appropriate rule. Defining a good selection rule is a difficult problem, which has been addressed through a heuristic criterion~\cite{PM}, random selection~\cite{MKR}, or $L^1$ regularization~\cite{HF}. 

We propose to construct a design tailored to the framework of an fGn, where observations are $(Y^H_0,Y^H_1,...,Y^H_{N-1})'$. Though we want to restrict the number of sub-vectors, we want all the observations to be used and thus to appear in at least one sub-vector. Therefore, we plan to build $N_{\text{cl}}\leq N$ vectors $\mathbf V_1$, ..., $\mathbf V_{N_{\text{cl}}}$, with the constraint that each vector $\mathbf V_k$ contains the observation $Y^H_{k-1}$. Other components of the vectors are to be determined by maximizing a given information. 

For the first vector, $\mathbf V_1$, we select, in addition to $t_{1,1}=0$, the times $t_{1,2},...,t_{1,p}$ maximizing the theoretical Fisher information, as expressed in Theorem~\ref{thm:Fisher_fGn}:
$$(t_{1,2},...,t_{1,p})=\argmax_{(t_2,...,t_p)\in\llbracket 1,N-1\rrbracket,t_2\neq ...\neq t_p} \mathcal I((Y^H_0,Y^H_{t_2},...,Y^H_{t_p})';H).$$
The optimization is a simple combinatorial optimization for $N=60$. We have chosen this small value of $N$ to limit the computational cost of the optimization, but it is not much restrictive since we know, from Section~\ref{sec:infoFG_Gaussian}, that distant observations lead to a small Fisher information. The result may depend as well on $H$ and $p$, but we have repeated the procedure for a large range of Hurst exponents in $[0,1]$ and for $p\in\llbracket 2,5\rrbracket$. We always get the same result, namely consecutive observations: $t_{1,i}=i-1$, for $i\in\llbracket 1,p\rrbracket$.

For the other vectors, we proceed iteratively. For example, for determining $\mathbf V_k=(Y^H_{t_{k,1}},Y^H_{t_{k,2}},...,Y^H_{t_{k,p}})'$, with $k\geq 2$ and $t_{k,1}=k-1$, we must account for the prior selection of $\mathbf V_1,...,\mathbf V_{k-1}$ and thus maximize the Godambe information, with the formula of Theorem~\ref{thm:GodambeGauss}, instead of the Fisher information:
$$(t_{k,2},...,t_{k,p})=\argmax_{(t_2,...,t_p)\in\llbracket 0,N-1\rrbracket\setminus\{k-1\} ,t_2\neq ...\neq t_p} \mathcal J(\mathbf V_1,...,\mathbf V_{k-1},(Y^H_{k-1},Y^H_{t_2},...,Y^H_{t_p})';H).$$
As suggested by Figure~\ref{fig:GodambeDecr}, the numerical optimization concludes that the optimal vector also consists in consecutive observations: $t_{k,i}=k+i-2$, for $i\in\llbracket 1,p\rrbracket$. We stop the selection process when all the observations are included in a vector, that is once we have selected $\mathbf V_1,...,\mathbf V_{N_{\text{cl}}}$, with $N_{\text{cl}}=N+1-p$.

Our method thus provides evidence that, for the fGn, selecting all the consecutive sub-vectors is optimal in terms of stepwise information, under specific constraints on $p$ and $N_{\text{cl}}$. In the same spirit, we also propose a more parsimonious design, with the partition of $(Y^H_0,Y^H_1,...,Y^H_{p\lfloor N/p\rfloor-1})'$ in $\lfloor N/p\rfloor$ sub-vectors of size $p$ containing consecutive observations. The selection of consecutive observations also provides a substantial computational advantage, as the covariance matrix involved in the composite likelihood is identical across all sub-vectors. As a result, its determinant and inverse need to be computed only once.

Using either the consecutive or the partition design, we give a brief overview of the asymptotic behavior of the composite maximum likelihood estimator, when $n$ and $N_{\text{cl}}$ tend to infinity, for a fixed value of $p$. The following argument is intended to show the role played by the Godambe information in this context rather than provide a complete proof. First, in the fGn case, the scores evaluated in the true value of the parameter are pair functions of the sub-vectors of observations, as one can see in equation~\eqref{eq:scoreTrace}. They also have a zero expectation. Therefore, their Hermite rank is at least 2. We can then apply the vector version of the Breuer-Major theorem~\cite[Theorem 4]{Arcones}\cite{NPP} and find that the normalized composite score $N_{\text{cl}}^{-1/2}d\mathcal C(H)/dH$ converges in distribution, when $H<3/4$, toward a Gaussian variable of variance $\E[(\partial_H\ell(\mathbf V_1;H,\sigma^2))^2] + 2\sum_{k=1}^{\infty}\E[\partial_H\ell(\mathbf V_1;H,\sigma^2)\partial_H\ell(\mathbf V_{k+1};H,\sigma^2)]$. Moreover, under suitable regularity conditions ensuring the uniform convergence of the Hessian, Birkhoff's ergodic theorem in our fGn framework yields that $(1/N_{\text{cl}})d^2\mathcal C(\widehat H_{N_{\text{cl}}})/dH^2$ converges uniformly almost surely toward $-\mathcal I(\mathbf V_1;H)$, for any sequence $\widehat H_{N_{\text{cl}}}$ converging toward $H$. Going back to the composite maximum likelihood estimator $\widehat H_{\mathcal C,N_{\text{cl}}}$, it can be seen as an M-estimator, so that, applying a first-order Taylor expansion to $d\mathcal C(\widehat H_{\mathcal C,N_{\text{cl}}})/dH$ around $H$, we get 
$$\sqrt{N_{\text{cl}}}(\widehat H_{\mathcal C,N_{\text{cl}}} - H) = -\left(\frac{1}{N_{\text{cl}}}\frac{d^2\mathcal C(\widetilde H_{\mathcal C,N_{\text{cl}}})}{dH^2}\right)^{-1} \frac{1}{\sqrt{N_{\text{cl}}}}\frac{d\mathcal C(H)}{dH},$$
where $\widetilde H_{\mathcal C,N_{\text{cl}}}$ lies between $\widehat H_{\mathcal C,N_{\text{cl}}}$ and $H$. Slutsky's theorem thus concludes that the above equation converges weakly toward a centered Gaussian distribution whose variance is the limit of $(\mathcal J(\mathbf V_1,...,\mathbf V_{N_{\text{cl}}};H)/N_{\text{cl}})^{-1}$ when $N_{\text{cl}}$ tends to infinity.

\subsection{Information for an fGn derived from an fBm with arbitrary time scale}\label{sec:derivfGN}

The moment method does not provide a natural way of estimating the Hurst exponent of a stationary transformation of fractal processes, because these processes are no longer self-similar~\cite{GarcinLamperti,GarcinEstimLamp}. Therefore, for estimating the Hurst exponent of an fGn with the moment method, it is recommended to work instead with the aggregated version of the process, namely with the underlying fBm~\cite{JLM}. However, we show below that this is only possible when the fGn is sampled with a time step of 1.

Starting from a continuous-time fGn $Y^H_t$, defined by $Y^H_t = X^H_{t+1} - X^H_t$, where $X^H_t$ is an fBm, we can consider vectors of discrete-time observations of various types. For example, $(Y^H_0,Y^H_1,...,Y^H_{N-1})'$ easily provides us with regularly sampled observations of the latent fBm: $X^H_1=Y^H_0$ and, $\forall i\in\llbracket 2,N\rrbracket$, $X^H_i=\sum_{j=0}^{i-1} Y^H_j$. But if we observe instead $(\widetilde Y^H_0,\widetilde Y^H_1,...,\widetilde Y^H_{N-1})'$, where 
$$\widetilde Y^H_t = Y^H_{t\delta} = X^H_{t\delta+1} - X^H_{t\delta},$$ 
then, depending on the value of $\delta$, it may not be possible to recover the latent fBm by aggregation. Indeed, when $\delta<1$, two consecutive observations $\widetilde Y^H_i$ and $\widetilde Y^H_{i+1}$ correspond to overlapping increments of the underlying fBm. If $\delta>1$, aggregating $\widetilde Y^H_0,\widetilde Y^H_1,...,\widetilde Y^H_{N-1}$ is not enough since gaps arise between increments of the fBm. Therefore, $\delta=1$ is the only case for which we can build a vector of values of the fBm.

By contrast, likelihood methods, including composite ones, remain valid for any $\delta>0$, since the vector $(\widetilde Y^H_0,\widetilde Y^H_1,...,\widetilde Y^H_{N-1})'$ is centered and  Gaussian, with a pairwise covariance that can be expressed as
$$\Cov\left(\widetilde Y^H_i, \widetilde Y^H_j\right) = \frac{\sigma^2}{2}\left(\left|(i-j)\delta-1\right|^{2H} - 2\left|(i-j)\delta\right|^{2H} +\left|(i-j)\delta+1\right|^{2H}\right).$$
This expression allows one to construct the covariance matrix $\sigma^2 \mathbf{R}$ and to extend the results of this paper to the case of arbitrary time scale. Moreover, since the entries of the covariance matrix depend on the lags $(i-j)$, the resulting covariance matrix is also Toeplitz and Trench algorithms can be used for classical or composite likelihood inferences. 

When $H<1/2$, $Y^H_i$ and $Y^H_{j}$ are negatively correlated when $i\neq j$. But, when considering $\delta<1$, because of the overlap, $\widetilde Y^H_i$ and $\widetilde Y^H_{j}$ may be positively correlated for $j$ close to $i$, whereas the correlation remains negative when $|j-i|$ is large. Therefore, for a practical application of composite likelihood methods, it is recommended not to focus only on pairs of consecutive observations, which may lead to an overestimation of $H$, but to consider larger sub-vectors. In what follows, we apply to the estimation of the process $\widetilde Y^H_t$ the sequential optimal design that we introduced in Section~\ref{sec:Design} for the standard fGn.

\section{Simulation study}\label{sec:simul}

We compare composite likelihood methods with baseline approaches, namely the moment and the likelihood estimators, as described in Sections~\ref{sec:moment} and~\ref{sec:mle}. Data are generated according to an fGn. The quality of each estimator applied to a discrete-time trajectory depends on the size $N$ of the vector of observations, on the Hurst exponent $H$, and, for the composite methods, on the size $p$ of the sub-vectors. For each set of parameters and each estimation method, we retrieve bias and variance thanks to 400 independent trajectories.

First, we select $N=500$. Figure~\ref{fig:Simul_H} displays, for various $H\in(0,1)$, the bias and variance of the baseline estimators along with the composite estimator defined on all the sub-vectors of $p=15$ consecutive observations. One can also define composite estimators on non-overlapping sub-vectors. The lower accuracy of such an estimator can be compensated by increasing the size of the sub-vectors (here up to $p=25$), which proves useful when the generating process exhibits long memory, as one can see in the variance graph of Figure~\ref{fig:Simul_H}. The likelihood approach, with no surprise, shows the lowest bias and variance, whatever $H$. For $H$ lower than about 0.7, the implemented composite methods have a lower variance than the moment-based method but, for greater values of $H$, the moment-based approach shows a surprisingly low variance. This is to be balanced by the huge bias of this last method, compared to the one of composite approaches, so that our composite methods outperform the moment-based estimator for all values of $H$, when one considers the mean squared error (MSE).

\begin{figure}[htbp]
	\centering
		\includegraphics[width=0.45\textwidth]{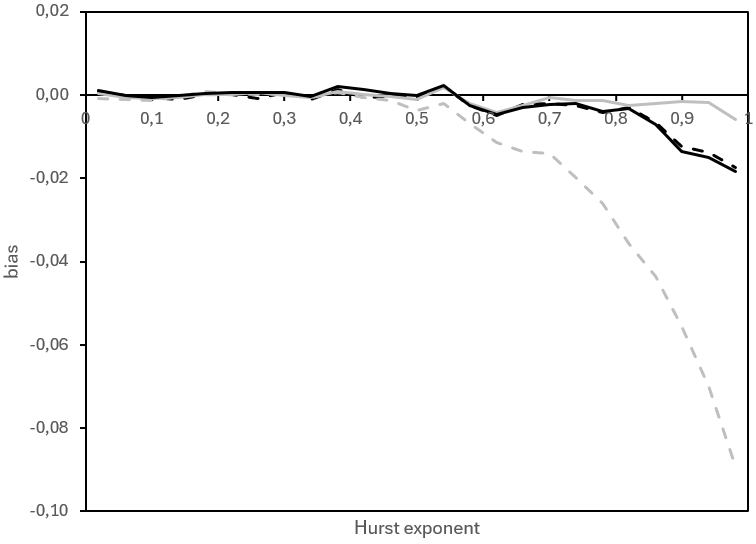}
		\includegraphics[width=0.45\textwidth]{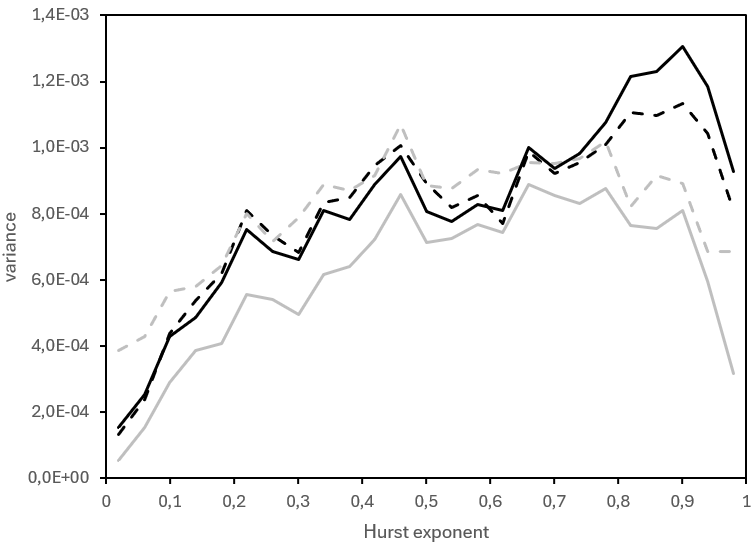} 
\begin{minipage}{0.7\textwidth}\caption{Bias (left graph) and variance (right graph) of the composite estimators with $p=15$ overlapping sub-vectors of consecutive observations (black solid line) or $p=25$ non-overlapping sub-vectors of consecutive observations (black dashed line), of the maximum likelihood estimator (grey solid line), and the moment-based estimator (grey dashed line), with $N=500$.}
	\label{fig:Simul_H}
\end{minipage}
\end{figure}

We did not include in Figure~\ref{fig:Simul_H} the second-order moment method, described in Section~\ref{sec:moment}. Indeed, despite a very small bias for all values of $H$, its variance is almost twice as large as that of the first-order estimator, resulting in an MSE approximately 1.8 times larger when $H\leq 0.7$. For $H\geq 0.8$, the second-order moment method outperforms the first-order estimator in terms of MSE; however, all the other methods shown in Figure~\ref{fig:Simul_H} still achieve substantially lower MSE values.

Though our implementation of the maximum likelihood estimator is accelerated using the Trench algorithm, this method remains slower than the others: 4 seconds for a single trajectory of 500 points, compared with 1 ms for the moment estimator, 264 ms (respectively 43 ms) for the composite likelihood estimator based on overlapping sub-vectors of size 15 (resp. non-overlapping sub-vectors of size 25). It is worth mentioning that a more classical composite estimator, namely the one defined on all possible pairs, requires a similar computational time (266 ms) while yielding a much larger MSE (on average $3.9\times 10^{-3}$ vs $8.6\times 10^{-4}$). 

Computational time is the main reason why the traditional maximum likelihood method cannot be used for large datasets. However, for large $N$, composite methods can still be compared with the moment estimator. Figure~\ref{fig:Simul_N} shows that the advantage of our composite estimators over moment-based estimators remains robust with respect to $N$.

\begin{figure}[htbp]
	\centering
		\includegraphics[width=0.45\textwidth]{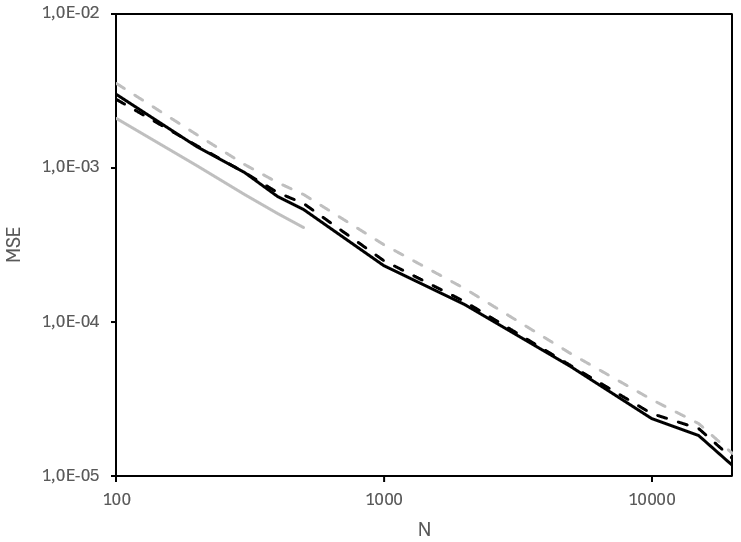}
		\includegraphics[width=0.45\textwidth]{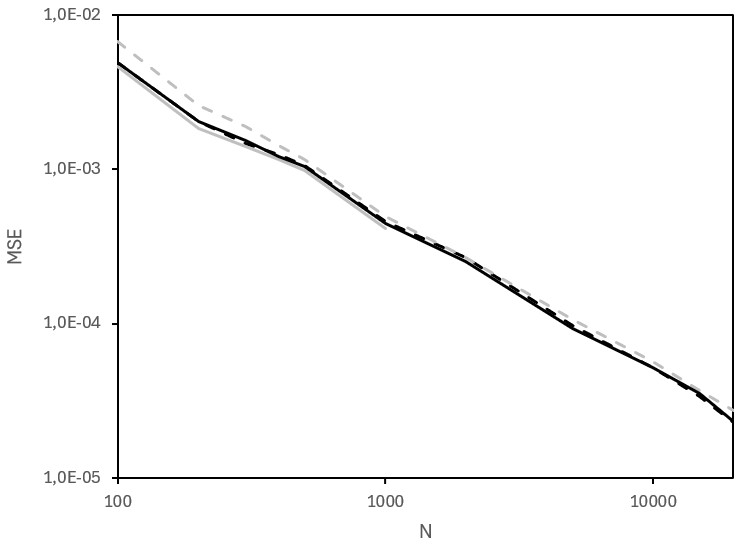} 
\begin{minipage}{0.7\textwidth}\caption{MSE, on a log-log scale, of the composite estimators with $p=15$ overlapping sub-vectors of consecutive observations (black solid line) or $p=25$ non-overlapping sub-vectors of consecutive observations (black dashed line), of the maximum likelihood estimator (grey solid line), and the moment-based estimator (grey dashed line), with $H=0.15$ (left graph) and $H=0.65$ (right graph).}
	\label{fig:Simul_N}
\end{minipage}
\end{figure}

The performance of composite estimators also depends on the size of the subsets. The larger $p$, the smaller the MSE, as one can see in Figure~\ref{fig:Simul_p}. Composite estimators become more accurate than moment estimators when $p$ is large enough: when $H=0.65$ it is from $p=4$ or $p=10$, whereas when $H=0.15$ it is from $p=15$ or $p=21$, depending on whether we consider overlapping or non-overlapping sub-vectors. The lower bound provided by the maximum likelihood estimator is reached asymptotically since it is a particular case of the composite method with $p=N$.

\begin{figure}[htbp]
	\centering
		\includegraphics[width=0.45\textwidth]{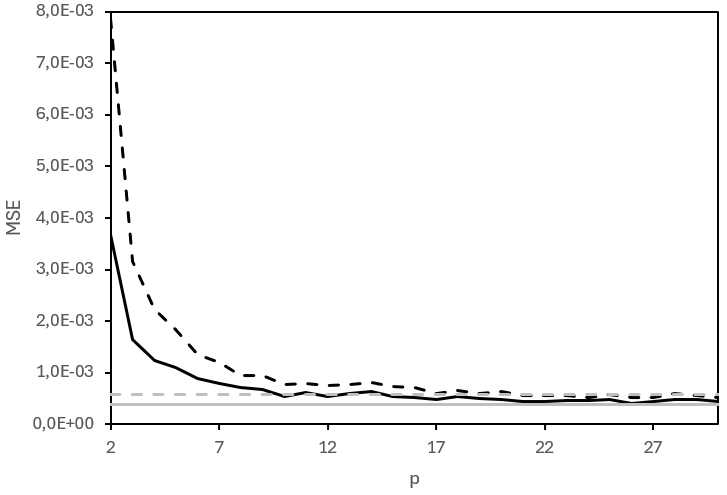}
		\includegraphics[width=0.45\textwidth]{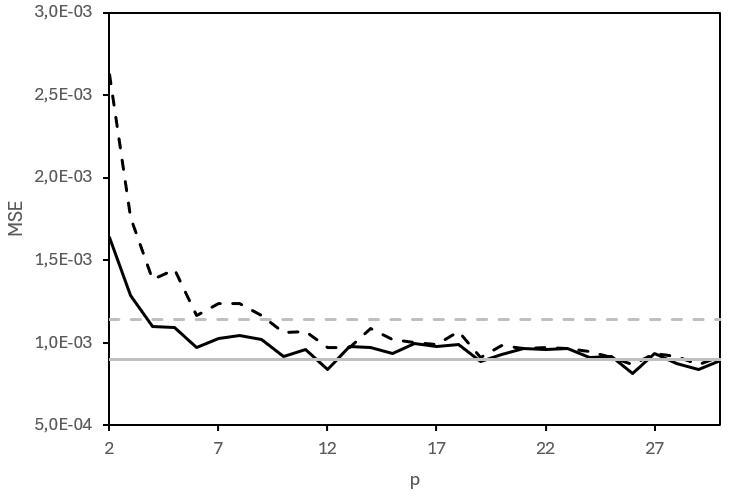} 
\begin{minipage}{0.7\textwidth}\caption{MSE of the composite estimators with overlapping (black solid line) or non-overlapping  (black dashed line) sub-vectors of consecutive observations, as a function of the size $p$ of each sub-vector, along with the MSE of the maximum likelihood estimator (grey solid line) and the moment-based estimator (grey dashed line), with $N=500$, $H=0.15$ (left graph), and $H=0.65$ (right graph).}
	\label{fig:Simul_p}
\end{minipage}
\end{figure}

\section{Empirical applications}\label{sec:appli}

We now consider several real-world datasets from financial markets and wind-speed measurements to illustrate the proposed method. For each daily sampled time series, we estimate the parameters of an fGn over a rolling window of 500 observations, using either the standard moment-based approach or the composite likelihood method. The latter is designed as introduced above, namely with all the possible sub-vectors of a fixed number of consecutive observations.

Since the true value of the $H$ parameter is unknown in empirical applications and since the model specification itself may be imperfect, the objective of this section is not to assess which method provides estimates closer to the true parameter. Instead, we investigate whether the two estimation procedures lead to different forecasting performances. If one method yields more accurate forecasts, this does not necessarily imply better inference of the true parameter; rather, it suggests that the method extracts predictive information from the data more effectively under the considered specification.

The forecasting method is consistent with the Gaussian fractional model~\cite{NP,GarcinForecast}, at a one-day horizon. The quality of the forecast is assessed with an MSE along with a Diebold-Mariano and West (DMW) test~\cite{DM,West} based on a squared-error loss. We use as well a hit ratio, defined as the proportion of good sign forecasts, together with a McNemar test.

\subsection{Volatility of financial markets}

We consider two kinds of volatility: the realized volatility of the S\&P 500 index, which is computed from historical price data sampled every 5 minutes in a single day, and the VIX, which is derived from the current price of a large set of options on the above stock index. The fBm is a widely used model for time series of log-volatility~\cite{CR}. Recent studies focus on moment-based estimation and find a Hurst parameter close to 0.1 for historical volatility~\cite{GJR,GG}, to 0.3 for volatility derived from a single option~\cite{LMP}, and to 0.4 for the VIX~\cite{ZC,GSV}.

The squared realized volatility series is imported from the formerly available Oxford-Man Institute of Quantitative Finance Realized Library and covers the period starting in January 2000 and finishing in November 2018. The VIX series is imported from Bloomberg over the period spanning from January 1996 to January 2026. We consider the logarithm of the two series and the composite likelihood approach is built on sub-vectors of 15 consecutive observations.

Gathering the estimates from all the windows of 500 observations, we infer a Hurst exponent of the realized volatility of 0.124 (respectively 0.129) in average, with standard deviation of 0.049 (resp. 0.063) for the composite likelihood approach (resp. the moment-based method). For the VIX, the estimated Hurst exponent is 0.423 (resp. 0.432) in average, with standard deviation of 0.025 (resp. 0.042). These values are consistent with the literature but it is worth mentioning that the composite likelihood approach seems to produce more stable estimates.

Table~\ref{tab:Vol} shows the results of the forecasting exercise. The two methods have close performance. The hit ratios do not differ significantly from each other, while the MSE indicates a better performance of the model estimated with the composite likelihood method.

\begin{table}[htbp]
\centering
\begin{tabular}{|l|c|c|c|c|c|}
\hline
\multirow{2}{*}{Series} & Estimation & \multicolumn{2}{c|}{hit ratio} & \multicolumn{2}{c|}{MSE} \\
\cline{3-6}
 & method & $\nu=5$ & $\nu=10$ & $\nu=5$ & $\nu=10$ \\
\hline
\multirow{2}{*}{RV} & Moments & $66.11\%$ & $66.35\%$ & 0.3610 & 0.3550 \\
\cline{2-6}
 & Composite L. & $66.20\%$ (0.394) & $66.49\%$ (0.303) & 0.3609 (-0.22) & 0.3545 (-1.12$^{\star}$) \\  
\hline
\multirow{2}{*}{VIX} & Moments & $52.23\%$ & $52.12\%$ & $4.946\times 10^{-3}$ & $4.939\times 10^{-3}$ \\
\cline{2-6}
 & Composite L. & $52.36\%$ (0.532) & $52.21\%$ (0.690) & $4.938\times 10^{-3}$ (-1.81$^{\star\star}$) & $4.928\times 10^{-3}$ (-2.12$^{\star\star\star}$) \\  
\hline
\end{tabular}
%\captionsetup{labelformat=empty}
\begin{minipage}{0.7\textwidth}\caption{For the realized volatility of the S\&P 500 index (RV) and the VIX, hit ratio and MSE with, in parenthesis, respectively the p-value of a McNemar test and the DMW statistic for the squared-error loss. The significance at $10\%$ ($^{\star}$), $5\%$ ($^{\star\star}$), and $1\%$ ($^{\star\star\star}$) is indicated for the DMW test. The forecasting model is an fGn, estimated either with moments or with a composite likelihood, and using $\nu\in\{5,10\}$ past daily observations for a forecast horizon of 1 day. Both the McNemar test and the DMW test are against the model estimated with moments and with the same $\nu$.}
\label{tab:Vol}
\end{minipage}
\end{table}

\subsection{Wind speed}

After the filtering of seasonal components, many models of wind speed assume stationary dynamics with serial dependence, typically within a fractional modeling framework. One can cite, for example, the fractional Ornstein-Uhlenbeck process~\cite{OADA,LLT}, the fGn~\cite{MS,HEAS}, or extensions like the generalized Cauchy process~\cite{LSZ}.

The data used in our study consist of 3-hour observations at Embrun (Hautes-Alpes, France), provided by M\'et\'eo-France between 1996 and 2024.\footnote{ \url{https://donneespubliques.meteofrance.fr/}, station 7591.} To reduce the seasonal effect, these observations are averaged at the daily scale. We model the logarithm of the series  with an fGn and the composite likelihood approach is built on sub-vectors of 25 consecutive observations. Moreover, for the composite likelihood method, we consider both the model in which observations are sample at a time step $\delta=1$ and the model with unknown $\delta$, as discussed in Section~\ref{sec:derivfGN}.

The estimated Hurst exponent is in average 0.740 for the moment method and 0.755 for the standard composite approach, with respective standard deviations of 0.036 and 0.031. When estimating also $\delta$, $H$ is 0.720 in average, with a standard deviation of 0.048, while $\delta$ is 0.94 in average, with a standard deviation of 0.20.

As one can see in Table~\ref{tab:Wind}, the moment method significantly outperforms the competing approaches, based on the hit ratio, only for $\nu=5$. In terms of MSE, the composite approach performs better regardless of $\nu$, although the difference is statistically significant only when estimating $\delta$.

\begin{table}[htbp]
\centering
\begin{tabular}{|l|c|c|c|c|}
\hline
Estimation & \multicolumn{2}{c|}{hit ratio} & \multicolumn{2}{c|}{MSE} \\
\cline{2-5}
method & $\nu=5$ & $\nu=10$ & $\nu=5$ & $\nu=10$ \\
\hline
Moments & $64.40\%$ & $64.57\%$ & 0.1767 & 0.1763 \\
\hline
Composite L. & $64.25\%$ (0.020) & $64.49\%$ (0.262) & 0.1766 (-0.16) & 0.1762 (-0.33) \\  
\hline
Composite L. $\delta$ & $64.12\%$ (0.196) & $64.71\%$ (0.505) & 0.1754 (-2.28$^{\star\star\star}$) & 0.1752 (-2.09$^{\star\star\star}$) \\  
\hline
\end{tabular}
%\captionsetup{labelformat=empty}
\begin{minipage}{0.7\textwidth}\caption{For the wind speed, hit ratio and MSE with, in parenthesis, respectively the p-value of a McNemar test and the DMW statistic for the squared-error loss. The significance at $10\%$ ($^{\star}$), $5\%$ ($^{\star\star}$), and $1\%$ ($^{\star\star\star}$) is indicated for the DMW test. The forecasting model is an fGn, estimated either with moments or with a composite likelihood (with time step $\delta$ fixed at 1 or estimated), and using $\nu\in\{5,10\}$ past daily observations for a forecast horizon of 1 day. Both the McNemar test and the DMW test are against the model estimated with moments and with the same $\nu$.}
\label{tab:Wind}
\end{minipage}
\end{table}

\section{Conclusion}\label{sec:concl}

As our simulations show, the composite likelihood method allows one to control the trade-off between estimation accuracy and computational complexity through the choice of the sub-vector size. Given this size, the key issue then becomes the design of the sub-vectors, namely how their components are selected. We have proposed an approach in which the optimality is seen as a stepwise maximization of the Godambe information. Introducing a theoretical expression of this information in the case of fractional Gaussian models, we were able to build a simple sequentially optimal design for the fBm and the fGn.

We have proposed as well a short application to real datasets. It must be seen as a motivation of such models and of the composite likelihood method. We know that multifractal extensions~\cite{BDM,SL}, or non-Gaussian ones, based on alpha-stable distributions~\cite{GSV} or even on jumps~\cite{AJC,BS}, are promising fractional alternatives of the fBm when modelling volatility. Similar multifractal~\cite{FF,CS} or alpha-stable~\cite{DSC} extensions are also considered for modelling wind speed. Interesting future research directions could thus consist in studying the sensitivity of the design to the above variations of the fBm.

\bibliographystyle{plain}
\bibliography{biblioEstimfBm}

\begin{thebibliography}{10}

\bibitem{AJC}
E.~Abi~Jaber and N.~De~Carvalho.
\newblock Reconciling rough volatility with jumps.
\newblock {\em S{IAM} journal on financial mathematics}, 15(3):785--823, 2024.

\bibitem{AW}
M.~Abt and W.J. Welch.
\newblock Fisher information and maximum-likelihood estimation of covariance
  parameters in {G}aussian stochastic processes.
\newblock {\em Canadian journal of statistics}, 26(1):127--137, 1998.

\bibitem{AG}
G.S. Ammar and W.B. Gragg.
\newblock Superfast solution of real positive definite {T}oeplitz systems.
\newblock {\em {SIAM} journal on matrix analysis and applications},
  9(1):61--76, 1988.

\bibitem{ACS}
M.~Anitescu, J.~Chen, and M.L. Stein.
\newblock An inversion-free estimating equations approach for {G}aussian
  process models.
\newblock {\em Journal of computational and graphical statistics},
  26(1):98--107, 2017.

\bibitem{Arcones}
M.A. Arcones.
\newblock Limit theorems for nonlinear functionals of a stationary {G}aussian
  sequence of vectors.
\newblock {\em Annals of probability}, pages 2242--2274, 1994.

\bibitem{BDM}
E.~Bacry, J.~Delour, and J.-F. Muzy.
\newblock Modelling financial time series using multifractal random walks.
\newblock {\em Physica {A}: statistical mechanics and its applications},
  299(1-2):84--92, 2001.

\bibitem{Beran}
J.~Beran.
\newblock {\em Statistics for long-memory processes}.
\newblock Routledge, 2017.

\bibitem{BeG}
M.~Bevilacqua and C.~Gaetan.
\newblock Comparing composite likelihood methods based on pairs for spatial
  {G}aussian random fields.
\newblock {\em Statistics and computing}, 25(5):877--892, 2015.

\bibitem{BS}
M.~Bibinger and M.~Sonntag.
\newblock Testing for jumps in processes with integral fractional part and
  jump-robust inference on the {H}urst exponent.
\newblock {\em To appear in Econometrics and statistics}, 2025.

\bibitem{BG}
X.~Brouty and M.~Garcin.
\newblock Fractal properties, information theory, and market efficiency.
\newblock {\em Chaos, solitons \& fractals}, 180:114543, 2024.

\bibitem{CS}
R.~Calif and F.G. Schmitt.
\newblock Modeling of atmospheric wind speed sequence using a lognormal
  continuous stochastic equation.
\newblock {\em Journal of wind engineering and industrial aerodynamics},
  109:1--8, 2012.

\bibitem{Cinkir}
Z.~Cinkir.
\newblock A fast elementary algorithm for computing the determinant of
  {T}oeplitz matrices.
\newblock {\em Journal of computational and applied mathematics}, 255:353--361,
  2014.

\bibitem{Coeurjolly2000}
J.-F. Coeurjolly.
\newblock Simulation and identification of the fractional {B}rownian motion: a
  bibliographical and comparative study.
\newblock {\em Journal of statistical software}, 5:1--53, 2000.

\bibitem{Coeurjolly2001}
J.-F. Coeurjolly.
\newblock Estimating the parameters of a fractional {B}rownian motion by
  discrete variations of its sample paths.
\newblock {\em Statistical inference for stochastic processes}, 4(2):199--227,
  2001.

\bibitem{CR}
F.~Comte and E.~Renault.
\newblock Long memory in continuous-time stochastic volatility models.
\newblock {\em Mathematical finance}, 8(4):291--323, 1998.

\bibitem{DM}
F.X. Diebold and R.S. Mariano.
\newblock Comparing predictive accuracy.
\newblock {\em Journal of business \& economic statistics}, 20(1):134--144,
  2002.

\bibitem{DSC}
S.~Duan, W.~Song, C.~Cattani, Y.~Yasen, and H.~Liu.
\newblock Fractional {L}\'evy stable and maximum {L}yapunov exponent for wind
  speed prediction.
\newblock {\em Symmetry}, 12(4):605, 2020.

\bibitem{FF}
K.~Falconer and C.~Fern\'andez.
\newblock Inference on fractal processes using multiresolution approximation.
\newblock {\em Biometrika}, 94(2):313--334, 2007.

\bibitem{Garcin2017}
M.~Garcin.
\newblock Estimation of time-dependent {H}urst exponents with variational
  smoothing and application to forecasting foreign exchange rates.
\newblock {\em Physica {A}: statistical mechanics and its applications},
  483:462--479, 2017.

\bibitem{GarcinLamperti}
M.~Garcin.
\newblock Hurst exponents and delampertized fractional {B}rownian motions.
\newblock {\em International journal of theoretical and applied finance},
  22(5):1950024, 2019.

\bibitem{GarcinEstimLamp}
M.~Garcin.
\newblock A comparison of maximum likelihood and absolute moments for the
  estimation of {H}urst exponents in a stationary framework.
\newblock {\em Communications in nonlinear science and numerical simulation},
  114:106610, 2022.

\bibitem{GarcinForecast}
M.~Garcin.
\newblock Forecasting with fractional {B}rownian motion: a financial
  perspective.
\newblock {\em Quantitative finance}, 22(8):1495--1512, 2022.

\bibitem{GG}
M.~Garcin and M.~Grasselli.
\newblock Long versus short time scales: the rough dilemma and beyond.
\newblock {\em Decisions in economics and finance}, 45(1):257--278, 2022.

\bibitem{GSV}
M.~Garcin, K.~Sawaya, and T.~Valade.
\newblock Prediction of linear fractional stable motions using codifference,
  with application to non-{G}aussian rough volatility.
\newblock {\em Preprint}, 2025.

\bibitem{GJR}
J.~Gatheral, T.~Jaisson, and M.~Rosenbaum.
\newblock Volatility is rough.
\newblock {\em Quantitative finance}, 18(6):933--949, 2018.

\bibitem{Harville}
D.A. Harville.
\newblock {\em Matrix algebra from a statistician's perspective}.
\newblock Springer, 1998.

\bibitem{HF}
Z.~Huang and D.~Ferrari.
\newblock Fast construction of optimal composite likelihoods.
\newblock {\em Statistica sinica}, 34(1):47--66, 2024.

\bibitem{HEAS}
S.~Hussain, A.~Elbergali, A.~Al-Masri, and G.~Shukur.
\newblock Parsimonious modelling, testing and forecasting of long-range
  dependence in wind speed.
\newblock {\em Environmetrics}, 15(2):155--171, 2004.

\bibitem{IL}
J.~Istas and G.~Lang.
\newblock Quadratic variations and estimation of the local {H}{\"o}lder index
  of a {G}aussian process.
\newblock {\em Annales de l'institut {H}enri {P}oincare ({B}) probability and
  statistics}, 33(4):407--436, 1997.

\bibitem{JHJ}
R.~Jennane, R.~Harba, and G.~Jacquet.
\newblock M{\'e}thodes d'analyse du mouvement brownien fractionnaire:
  th{\'e}orie et r{\'e}sultats comparatifs.
\newblock {\em Traitement du signal}, 18(5-6):419--436, 2001.

\bibitem{JLM}
H.-D.J. Jeong, J.-S.R. Lee, D.~McNickle, and K.~Pawlikowski.
\newblock Comparison of various estimators in simulated {FGN}.
\newblock {\em Simulation modelling practice and theory}, 15(9):1173--1191,
  2007.

\bibitem{JL}
H.~Joe and Y.~Lee.
\newblock On weighting of bivariate margins in pairwise likelihood.
\newblock {\em Journal of multivariate analysis}, 100(4):670--685, 2009.

\bibitem{KMS}
H.~Khalil, B.~Mourrain, and M.~Schatzman.
\newblock Superfast solution of {T}oeplitz systems based on syzygy reduction.
\newblock {\em Linear algebra and its applications}, 438(9):3563--3575, 2013.

\bibitem{LLT}
S.C. Lim, M.~Li, and L.P. Teo.
\newblock Langevin equation with two fractional orders.
\newblock {\em Physics letters {A}}, 372(42):6309--6320, 2008.

\bibitem{Lindsay}
B.G. Lindsay.
\newblock Composite likelihood methods.
\newblock In {\em Statistical inference from stochastic processes}, volume~80
  of {\em Contemporary mathematics}, pages 221--239. American mathematical
  society, 1988.

\bibitem{LYS}
B.G. Lindsay, G.Y. Yi, and J.~Sun.
\newblock Issues and strategies in the selection of composite likelihoods.
\newblock {\em Statistica sinica}, 21(1):71--105, 2011.

\bibitem{LSZ}
H.~Liu, W.~Song, and E.~Zio.
\newblock Generalized {C}auchy difference iterative forecasting model for wind
  speed based on fractal time series.
\newblock {\em Nonlinear dynamics}, 103:759--773, 2021.

\bibitem{LMP}
G.~Livieri, S.~Mouti, A.~Pallavicini, and M.~Rosenbaum.
\newblock Rough volatility: evidence from option prices.
\newblock {\em IISE transactions}, 50(9):767--776, 2018.

\bibitem{MvN}
B.B. Mandelbrot and J.W. van Ness.
\newblock Fractional {B}rownian motions, fractional noises and applications.
\newblock {\em {SIAM} review}, 10(4):422--437, 1968.

\bibitem{MKR}
G.~Mazo, D.~Karlis, and A.~Rau.
\newblock A randomized pairwise likelihood method for complex statistical
  inferences.
\newblock {\em Journal of the American statistical association},
  119(547):2317--2327, 2024.

\bibitem{MS}
M.M. Meerschaert and F.~Sabzikar.
\newblock Tempered fractional {B}rownian motion.
\newblock {\em Statistics \& probability letters}, 83(10):2269--2275, 2013.

\bibitem{NJK}
C.T. Ng, H.~Joe, D.~Karlis, and J.~Liu.
\newblock Composite likelihood for time series models with a latent
  autoregressive process.
\newblock {\em Statistica Sinica}, 21(1):279--305, 2011.

\bibitem{NPP}
I.~Nourdin, G.~Peccati, and M.~Podolskij.
\newblock Quantitative {B}reuer--{M}ajor theorems.
\newblock {\em Stochastic processes and their applications}, 121(4):793--812,
  2011.

\bibitem{NP}
C.J. Nuzman and H.V. Poor.
\newblock Linear estimation of self-similar processes via {L}amperti's
  transformation.
\newblock {\em Journal of applied probability}, 37(2):429--452, 2000.

\bibitem{OADA}
S.~Obukhov, E.M. Ahmed, D.Y. Davydov, T.~Alharbi, A.~Ibrahim, and Z.M. Ali.
\newblock Modeling wind speed based on fractional {O}rnstein-{U}hlenbeck
  process.
\newblock {\em Energies}, 14(17):5561, 2021.

\bibitem{PM}
I.~Papageorgiou and I.~Moustaki.
\newblock Sampling of pairs in pairwise likelihood estimation for latent
  variable models with categorical observed variables.
\newblock {\em Statistics and computing}, 29(2):351--365, 2019.

\bibitem{PT}
M.A. Poletti and P.D. Teal.
\newblock A superfast {T}oeplitz matrix inversion method for single- and
  multi-channel inverse filters and its application to room equalization.
\newblock {\em {IEEE}/{ACM} transactions on audio, speech, and language
  processing}, 29:3144--3157, 2021.

\bibitem{SL}
C.~Sattarhoff and T.~Lux.
\newblock Forecasting the variability of stock index returns with the
  multifractal random walk model for realized volatilities.
\newblock {\em International journal of forecasting}, 39(4):1678--1697, 2023.

\bibitem{SYZ}
S.~Shi, J.~Yu, and C.~Zhang.
\newblock Fractional {G}aussian noise: {S}pectral density and estimation
  methods.
\newblock {\em Journal of time series analysis}, 46(6):1146--1174, 2025.

\bibitem{SCA}
M.L. Stein, J.~Chen, and M.~Anitescu.
\newblock Stochastic approximation of score functions for {G}aussian processes.
\newblock {\em Annals of applied statistics}, 7(2):1162--1191, 2013.

\bibitem{Trench}
W.F. Trench.
\newblock An algorithm for the inversion of finite {T}oeplitz matrices.
\newblock {\em Journal of the society for industrial and applied mathematics},
  12(3):515--522, 1964.

\bibitem{VRF}
C.~Varin, N.~Reid, and D.~Firth.
\newblock An overview of composite likelihood methods.
\newblock {\em Statistica sinica}, 21(1):5--42, 2011.

\bibitem{VV}
C.~Varin and P.~Vidoni.
\newblock A note on composite likelihood inference and model selection.
\newblock {\em Biometrika}, 92(3):519--528, 2005.

\bibitem{West}
K.D. West.
\newblock Asymptotic inference about predictive ability.
\newblock {\em Econometrica: journal of the econometric society},
  64(5):1067--1084, 1996.

\bibitem{XRX}
L.~Xu, N.~Reid, and X.~Xu.
\newblock A note on information bias and efficiency of composite likelihood.
\newblock {\em Statistica sinica}, 34(1):523--526, 2024.

\bibitem{ZC}
Q.~Zhao and A.~Chronopoulou.
\newblock A new proxy for estimating the roughness of volatility.
\newblock {\em Journal of risk and financial management}, 17(4):131, 2024.

\bibitem{Zohar}
S.~Zohar.
\newblock Toeplitz matrix inversion: {T}he algorithm of {WF} {T}rench.
\newblock {\em Journal of the {ACM}}, 16(4):592--601, 1969.

\end{thebibliography}

\appendix

\section{Proof of Theorem~\ref{thm:Fisher_fGn}}\label{sec:proof_Fisher_fGn}

Before proving Theorem~\ref{thm:Fisher_fGn}, we introduce a useful lemma.

\begin{lem}\label{lem:deriv_g}
Let $H\in(0,1)$, $\tau>1$, and $g$ and $h$ be the functions defined, for $x\in[0,1/\tau]$, by
$$\left\{\begin{array}{ccl}
g(x) & = & (1+x)^{2H} - 2 + (1-x)^{2H} \\
h(x) & = & (1-x)^{2H}\ln{(1-x)} + (1+x)^{2H}\ln{(1+x)}.
\end{array}\right.$$
Then, their $n$-th derivative, for $n\in\mathbb N\setminus\{0\}$, is 
\begin{equation}\label{eq:systLemma_gh}
\left\{\begin{array}{ccl}
g^{(n)}(x) & = & \left((1+x)^{2H-n} + (-1)^n(1-x)^{2H-n}\right)A_n \\
h^{(n)}(x) & = &  (1+x)^{2H-n}(A_n\ln(1+x)+B_n)+(-1)^n(1-x)^{2H-n}(A_n\ln(1-x)+B_n),
\end{array}\right.
\end{equation}
where $A_n$ and $B_n$ are defined as in Theorem~\ref{thm:Fisher_fGn}, equation~\eqref{eq:FisherGauss_AB}. Moreover, if $n$ is odd, we have
$$\left\{\begin{array}{ccl}
|g^{(n)}(x)| & \le & \alpha_{g, n}x \\
|h^{(n)}(x)| & \le & \alpha_{h, n}x,
\end{array}\right.$$
with $\alpha_{g, n}$ and $\alpha_{h, n}$ like in Theorem~\ref{thm:Fisher_fGn}.
\end{lem}

\begin{proof}
We first prove equation~\eqref{eq:systLemma_gh} by induction. A straightforward differentiation provides us with the case $n=1$. Differentiating the $n$-th derivative then gives
$$
g^{(n+1)}(x) =A_n(2H-n) \left((1+x)^{2H-(n+1)}+(-1)^{n+1}(1-x)^{2H-(n+1)}\right)
$$
and
$$\begin{array}{ccl}
h^{(n+1)}(x) & = & (1+x)^{2H-(n+1)}\left((2H-n)A_n\ln(1+x)+(2H-n)B_n+A_n\right) \\
 & & + (-1)^{n+1}(1-x)^{2H-(n+1)}\left((2H-n)A_n\ln(1-x)+(2H-n)B_n+A_n\right).
\end{array}$$
Noting the definition of the sequences $A_n$ and $B_n$ in equation~\eqref{eq:FisherGauss_AB}, we get an expression consistent with equation~\eqref{eq:systLemma_gh}, thus proving the induction.

When $n$ is odd, we apply the mean value theorem to $y\mapsto y^{2H-n}A_n$ and $y\mapsto y^{2H-n}(A_n\ln(y)+B_n)$ on $[1-x,1+x]$: there exists $c,c'\in(1-x,1+x)\subset(1-\tau^{-1},1+\tau^{-1})$ such that
$$\left\{\begin{array}{ccl}
g^{(n)}(x) & = & 2xA_{n+1}c^{2H-n-1} \\
h^{(n)}(x) & = & 2x(A_{n+1}\ln(c')+B_{n+1})(c')^{2H-n-1}.
\end{array}\right.$$
Hence
$$\left\{\begin{array}{ccl}
|g^{(n)}(x)| & \le & 2|A_{n+1}|(1-\tau^{-1})^{2H-n-1}x = \alpha_{g,n}x \\
|h^{(n)}(x)| & \le & 2x\left(|A_{n+1}\ln(1-\tau^{-1})| +|B_{n+1}|\right)(1-\tau^{-1})^{2H-n-1} = \alpha_{h,n}x.
\end{array}\right.$$
\end{proof}

We can now prove Theorem~\ref{thm:Fisher_fGn}.

\begin{proof}
We note that the inverse and the derivative of the correlation matrix $\mathbf R$ of $Y_{t}$ and $Y_{t+\tau}$ are respectively 
$$\mathbf R^{-1}=\frac{1}{1-\rho_{H,\tau}^2}\left(\begin{array}{cc}
1 & -\rho_{H,\tau} \\
-\rho_{H,\tau} & 1
\end{array}\right) \ \text{and} \ \mathbf R_H=\left(\begin{array}{cc}
0 & \partial_H\rho_{H,\tau} \\
\partial_H\rho_{H,\tau} & 0
\end{array}\right).$$
Using these two expressions along with equation~\eqref{eq:infoFisher_trace}, we directly get equation~\eqref{eq:Fisher_Gauss}.

In the fGn case, we have $\rho_{H,\tau}=\tau^{2H}g(1/\tau)/2$ and $\partial_H \rho_{H,\tau}=\tau^{2H}\left[\ln(\tau)g(1/\tau)+h(1/\tau)\right]$, where $g$ and $h$ are defined as in the statement of Lemma~\ref{lem:deriv_g}. Noting that, after Lemma~\ref{lem:deriv_g}, we have $g(0)=g'(0)=0$, the second-order Taylor expansion of $g$ simply writes
$$g(x)= 2H(2H-1)x^2 + \int_0^x \frac{g^{(3)}(t)}{2}(x-t)^2\,dt.$$
Similarly, using again Lemma~\ref{lem:deriv_g}, we get
$$h(x)=(4H-1)x^2+ \int_0^x \frac{h^{(3)}(t)}{2}(x-t)^2\,dt.$$
We can now write the expansion $\rho_{\tau}=a_{\tau}+r_{\tau}$ and $\partial_H\rho_{\tau}=b_{\tau}+s_{\tau}$, with the remainders 
$$\left\{\begin{array}{ccl}
r_{\tau} & = & \frac{\tau^{2H}}{4}\int_0^{1/\tau}g^{(3)}(t)\left(\frac1\tau-t\right)^2\,dt \\
s_{\tau} & = & 2r_{\tau}\ln \tau + \frac{\tau^{2H}}{2}\int_0^{1/\tau}h^{(3)}(t)\left(\frac1\tau-t\right)^2\,dt.
\end{array}\right.$$
From Lemma~\ref{lem:deriv_g}, we have the bounds $|g^{(3)}(t)| \le \alpha_{g,3}\, t$ and $|h^{(3)}(t)| \le \alpha_{h,3}\, t$, for $t\in[0,1/\tau]$, so that
\begin{equation}\label{eq:borne_r_th1}
|r_{\tau}|\leq \frac{\tau^{2H}}{4}\alpha_{g,3}
\int_0^{1/\tau}t\left(\frac1\tau-t\right)^2\,dt =
\frac{\alpha_{g,3}}{48}\,\tau^{2H-4},
\end{equation}
and, using the same kind of calculation along with the triangle inequality,
$$|s_{\tau}|\leq\frac{\alpha_{g,3}}{24}\tau^{2H-4}\ln(\tau) + \frac{\alpha_{h,3}}{24}\tau^{2H-4}=\tau^{2H-4}M_{\tau}.$$

Going back to the expression of the Fisher information and introducing $F(x)=(1+x^2)/(1-x^2)^2$, we can insert the above expansions of $\rho_{H,\tau}$ and $\partial_H \rho_{H,\tau}$:
$$\mathcal I(Y_t,Y_{t+\tau};H)=(b_{\tau}+s_{\tau})^2F(a_{\tau}+r_{\tau}).$$
We want to approximate it with $b_{\tau}^2F(a_{\tau})$ and need to bound the remainder $\mathcal R_{\tau}=\mathcal I(Y_t,Y_{t+\tau};H)-b_{\tau}^2F(a_{\tau})$. By the triangle inequality, we have
\begin{equation}\label{eq:proofInfoBorneFGN}
|\mathcal R_{\tau}|
\le
|s_{\tau}|(2|b_{\tau}|+|s_{\tau}|)F(a_{\tau}+r_{\tau})
+
b_{\tau}^2\,|F(a_{\tau}+r_{\tau})-F(a_{\tau})|.
\end{equation}
We note that $F$ is an increasing function in $[0,1)$, with a continuous and increasing derivative. Therefore, noting 
$$q_{\tau}=|a_{\tau}|+\left|\frac{\alpha_{g,3}}{48}\,\tau^{2H-4}\right|,$$
we have $|a_{\tau}+r_{\tau}|\leq|a_{\tau}|+|r_{\tau}|\leq q_{\tau}$, according to equation~\eqref{eq:borne_r_th1}, and $F(a_{\tau}+r_{\tau})\le F(q_{\tau})$ as well as, using a Taylor expansion, $|F(a_{\tau}+r_{\tau})-F(a_{\tau})| \le |r_{\tau}|F'(q_{\tau})$. Combining these two inequalities with formula~\eqref{eq:proofInfoBorneFGN} and with $F'(x)=2x(3+x^2)/(1-x^2)^3$ leads, after factorizing by $\tau^{2H-4}$, to the bound of $\mathcal R_{\tau}$ provided in Theorem~\ref{thm:Fisher_fGn}.
\end{proof}

\section{Proof of Theorem~\ref{thm:GodambeGauss}}\label{sec:proof_GodambeGauss}

Before proving the theorem, we need the following lemma.

\begin{lem}\label{lem:EVAVVBV}
Let $\mathbf V_1$ and $\mathbf V_2$ be two vectors of Gaussian variables, of size respectively $p_1$ and $p_2$, of covariance $\sigma^2\mathbf{R}_1$ and $\sigma^2\mathbf{R}_2$, and of cross-covariance matrix $\sigma^2\bm\Lambda_{1,2}=\sigma^2\bm\Lambda_{2,1}'$. Considering we are given $\mathbf A$ and $\mathbf B$, symmetric matrices of size $p_1\times p_1$ and $p_2\times p_2$, we have
$$\E\left[\mathbf V_1'\mathbf A\mathbf V_1\mathbf V_2'\mathbf B\mathbf V_2\right]=\sigma^4\tr(\mathbf A\mathbf{R}_1)\tr(\mathbf B\mathbf{R}_2)+2\sigma^4\tr(\mathbf A\bm\Lambda_{1,2}\mathbf B\bm\Lambda_{2,1}).$$
In the particular case where $\mathbf V_1=\mathbf V_2$, we get
$$\E\left[\mathbf V_1'\mathbf A\mathbf V_1\mathbf V_1'\mathbf B\mathbf V_1\right]=\sigma^4\tr(\mathbf A\mathbf{R}_1)\tr(\mathbf B\mathbf{R}_1)+2\sigma^4\tr(\mathbf A\mathbf{R}_1 \mathbf B\mathbf{R}_1).$$
\end{lem}

\begin{proof}
We write $\mathbf V_1$ and $\mathbf V_2$ as sub-vectors of a Gaussian vector $\mathbf V$ of covariance $\sigma^2\mathbf{R}$: $\mathbf V_1=\mathbf S_1\mathbf V$ and $\mathbf V_2=\mathbf S_2\mathbf V$. We can write the Cholesky decomposition $\mathbf{R}=\mathbf L\mathbf L'$, so that $\mathbf V=\sigma\mathbf L\mathbf U$, where $\mathbf U=(U_1,...,U_n)$ is a vector of $n$ i.i.d. standard Gaussian variables. Expanding the following matrix product, we get
$$\mathbf U'\mathbf A\mathbf U\mathbf U'\mathbf B\mathbf U=\sum_{i,j}U_iU_j(\mathbf A)_{ij}\sum_{k,l}U_kU_l(\mathbf B)_{kl},$$
so that, by linearity
$$\E\left[\mathbf U'\mathbf A\mathbf U\mathbf U'\mathbf B\mathbf U\right]=\sum_{i,j,k,l}(\mathbf A)_{ij}(\mathbf B)_{kl}\E\left[U_iU_jU_kU_l\right].$$
Thanks to the i.i.d. standard Gaussian assumption, we have 
$$\E\left[U_iU_jU_kU_l\right]=\left\{\begin{array}{llc}
3 & \text{if} & i=j=k=l \\
1 & \text{if} & i=j\neq k=l \\
 & \text{or if} & i=k\neq j=l \\
 & \text{or if} & i=l\neq j=k \\
0 & \text{else.} & 
\end{array}\right.$$
Therefore,
$$\begin{array}{ccl}
\E\left[\mathbf U'\mathbf A\mathbf U\mathbf U'\mathbf B\mathbf U\right] & = & 3\sum_i (\mathbf A)_{ii}(\mathbf B)_{ii} + \sum_{i,k\neq i} (\mathbf A)_{ii}(\mathbf B)_{kk} + \sum_{i,j\neq i} (\mathbf A)_{ij}(\mathbf B)_{ij} + \sum_{i,j\neq i} (\mathbf A)_{ij}(\mathbf B)_{ji} \\
 & = & \sum_{i,k} (\mathbf A)_{ii}(\mathbf B)_{kk} + 2\sum_{i,j} (\mathbf A)_{ij}(\mathbf B)_{ji} \\
 & = & \tr(\mathbf A)\tr(\mathbf B)+2\tr(\mathbf A\mathbf B),
\end{array}$$
where we used the symmetry of $B$ in the second line. Going back to $\mathbf V_1=\sigma\mathbf S_1\mathbf L\mathbf U$ and $\mathbf V_2=\sigma\mathbf S_2\mathbf L\mathbf U$, we have
$$\begin{array}{ccl}
\E\left[\mathbf V_1'\mathbf A\mathbf V_1\mathbf V_2'\mathbf B\mathbf V_2\right] & = & \sigma^4\E\left[\mathbf U'\mathbf L'\mathbf S_1'\mathbf A\mathbf S_1\mathbf L\mathbf U\mathbf U'\mathbf L'\mathbf S_2'\mathbf B\mathbf S_2\mathbf L\mathbf U\right] \\
 & = & \sigma^4\tr(\mathbf L'\mathbf S_1'\mathbf A\mathbf S_1\mathbf L)\tr(\mathbf L'\mathbf S_2'\mathbf B\mathbf S_2\mathbf L)+2\sigma^4\tr(\mathbf L'\mathbf S_1'\mathbf A\mathbf S_1\mathbf L\mathbf L'\mathbf S_2'\mathbf B\mathbf S_2\mathbf L) \\
 & = & \sigma^4\tr(\mathbf A\mathbf S_1\mathbf{R} \mathbf S_1')\tr(\mathbf B\mathbf S_2\mathbf{R} \mathbf S_2')+2\sigma^4\tr(\mathbf A\mathbf S_1\mathbf{R} \mathbf S_2'\mathbf B\mathbf S_2\mathbf{R} \mathbf S_1').
\end{array}$$
Noting that $\mathbf{R}_1=\mathbf S_1\mathbf{R} \mathbf S_1'$, $\mathbf{R}_2=\mathbf S_2\mathbf{R} \mathbf S_2'$, $\bm\Lambda_{1,2}=\mathbf S_1\mathbf{R} \mathbf S_2'$, and that $\bm\Lambda_{2,1}=\mathbf S_2\mathbf{R} \mathbf S_1'$, we get the expression given in the Lemma.
\end{proof}

We can now prove Theorem~\ref{thm:GodambeGauss}.

\begin{proof}
The covariance between two scores (or derivatives of $\ell$ with respect to $H$ only) is, thanks to the score formula provided in equation~\eqref{eq:scorefinal_trace},
\begin{equation}\label{eq:traceProduitScores}
\begin{array}{ccl}
\E\left[u(\mathbf V_j;H)u(\mathbf V_k;H)\right] & = & \frac{1}{4}\tr(\mathbf{R}_j^{-1}\mathbf{R}_{j,H})\tr(\mathbf{R}_k^{-1}\mathbf{R}_{k,H}) \\
& & -\frac{1}{4\sigma^2}\tr(\mathbf{R}_j^{-1}\mathbf{R}_{j,H})\E\left[\mathbf V_k'\mathbf{R}_k^{-1}\mathbf{R}_{k,H}\mathbf{R}_k^{-1}\mathbf V_k\right] \\
 & & -\frac{1}{4\sigma^2}\tr(\mathbf{R}_k^{-1}\mathbf{R}_{k,H})\E\left[\mathbf V_j'\mathbf{R}_j^{-1}\mathbf{R}_{j,H}\mathbf{R}_j^{-1}\mathbf V_j\right] \\
  & & +\frac{1}{4\sigma^4}\E\left[\mathbf V_j'\mathbf{R}_j^{-1}\mathbf{R}_{j,H}\mathbf{R}_j^{-1}\mathbf V_j\mathbf V_k'\mathbf{R}_k^{-1}\mathbf{R}_{k,H}\mathbf{R}_k^{-1}\mathbf V_k\right] \\
 & = & \frac{1}{2}\tr(\mathbf{R}_j^{-1}\mathbf{R}_{j,H}\mathbf{R}_j^{-1}\bm\Lambda_{j,k}\mathbf{R}_k^{-1}\mathbf{R}_{k,H}\mathbf{R}_k^{-1}\bm\Lambda_{k,j}),
\end{array}
\end{equation}
where we used equation~\eqref{eq:scoreVMV} and Lemma~\ref{lem:EVAVVBV}. In the particular case where $j=k$, we have $\bm\Lambda_{j,k}=\mathbf R_j=\mathbf R_k$ and the above expression with the trace is equal to the variance of the score and to the Fisher information provided in equation~\eqref{eq:infoFisher_trace}. Finally, the Godambe information with respect to the parameter $H$, which we give in Theorem~\ref{thm:GodambeGauss}, is obtained by inserting equation~\eqref{eq:traceProduitScores} into equation~\eqref{eq:GodambeUnSeulParam}.

In the stationary case where the vectors are of size 2, with the same time lag $\tau$ between their two observations, the numerator of the Godambe information becomes, using equation~\eqref{eq:infoFisher_trace} together with Theorem~\ref{thm:Fisher_fGn}:
$$\left(\sum_{k=1}^{N_{cl}}
\operatorname{tr}
(\mathbf{R}_k^{-1}\mathbf{R}_{k,H}\mathbf{R}_k^{-1}\mathbf{R}_{k,H})\right)^2 = 4N_{cl}^2\frac{(\partial_H\rho_{H,\tau})^{4}(1+\rho_{H,\tau}^{2})^2}{(1 - \rho_{H,\tau}^{2})^4}.$$
Regarding the denominator of the Godambe information, noting that 
$$\mathbf{\Lambda}_{j,k}=
\begin{pmatrix}
\rho_{H,t_j-t_k} & \rho_{H,t_j-t_k - \tau} \\
\rho_{H,t_j-t_k +\tau} & \rho_{H,t_j-t_k}
\end{pmatrix}$$
and that, for all $k$, 
$$\mathbf{R}_k^{-1}\mathbf{R}_{k, H}\mathbf{R}_k^{-1}
=
\frac{\partial_H \rho_{H,\tau}}{(1-\rho_{H,\tau}^{\,2})^2}
\begin{pmatrix}
-2\rho_{H,\tau} & 1+ \rho^2_{H,\tau} \\
1+ \rho^2_{H,\tau} & -2\rho_{H,\tau}
\end{pmatrix},$$
we obtain
$$\mathcal{J}\left((\mathbf{V_k})_{k\in\llbracket 1, N_{cl}\rrbracket} ; H\right)=\frac{2 N_{cl}^2(\partial_H \rho_{H,\tau})^2
(1+\rho_{H,\tau}^2)^2}{\sum_{j=1}^{N_{\mathrm{cl}}}\sum_{k=1}^{N_{\mathrm{cl}}}T_{t_j-t_k,H,\tau}}.$$
Finally, since $T_{\theta,H,\tau}=T_{-\theta,H,\tau}$ for all $\theta$, we get equation~\eqref{eq:Godambe_p2}. 
\end{proof}

\end{document}